\documentclass[11pt,letterpaper]{llncs}

\usepackage{float}
\usepackage{fullpage}
\usepackage{graphicx}
\usepackage{xcolor}
\usepackage{amsmath}
\usepackage{amssymb}
\usepackage[sort,nocompress]{cite}
\usepackage[ruled,vlined,lined,boxed,commentsnumbered]{algorithm2e}
\usepackage{todonotes}
\usepackage{verbatim}

\usepackage{pgf}
\usepackage{tikz}
\usetikzlibrary{arrows,automata}

\begin{document}

\title{Regular Languages meet Prefix Sorting\thanks{We wish to thank Travis Gagie for introducing us to the problem and for stimulating discussions.}}

\author{Jarno Alanko\inst{1} \and Giovanna D'Agostino\inst{2} \and Alberto Policriti\inst{2} \and Nicola Prezza\inst{3}\thanks{Corresponding author. Supported by the project MIUR-SIR CMACBioSeq (``Combinatorial methods for analysis and compression of biological sequences'') grant n.~RBSI146R5L.}}
\institute{University of Helsinki, Finland, \email{jarno.alanko@helsinki.fi}
\and 
University of Udine, Italy,
\email{\{giovanna.dagostino,alberto.policriti\}@uniud.it}
\and
University of Pisa, Italy, \email{nicola.prezza@di.unipi.it}
}

\maketitle

\begin{abstract}
Indexing strings via prefix (or suffix) sorting is, arguably, one of the most successful algorithmic techniques developed in the last decades. Can indexing be extended to languages? The main contribution of this paper is to initiate the study of the sub-class of regular languages accepted by an automaton whose states can be prefix-sorted. Starting from the recent notion of \emph{Wheeler graph} [Gagie et al., TCS 2017]---which extends naturally the concept of prefix sorting to labeled graphs---we investigate the properties of \emph{Wheeler languages}, that is, regular languages admitting an accepting Wheeler finite automaton. Interestingly, we characterize this family as the natural extension of regular languages endowed with the co-lexicographic ordering: when sorted, the strings belonging to a Wheeler language are partitioned into a \emph{finite} number of co-lexicographic \emph{intervals}, each formed by elements from a single Myhill-Nerode equivalence class. We proceed by proving several results related to Wheeler automata:
\begin{enumerate}
    \item[(i)] 
    We show that every Wheeler NFA (WNFA) with $n$ states admits an equivalent Wheeler DFA (WDFA) with at most $2n-1-|\Sigma|$ states ($\Sigma$ being the alphabet) that can be computed in $O(n^3)$ time. This is in sharp contrast with general NFAs
    (where the blow-up could be exponential).
    \item[(ii)] We describe 
    a quadratic algorithm to prefix-sort a proper superset of the WDFAs, a $O(n\log n)$-time \emph{online} algorithm to sort acyclic WDFAs, and an optimal linear-time offline algorithm to sort general WDFAs. By contribution (i), our algorithms can also be used to index \emph{any} WNFA at the moderate price of doubling the automaton's size.
    \item[(iii)] We provide a minimization theorem that characterizes the smallest WDFA recognizing the same language of any input WDFA. The corresponding constructive algorithm runs in optimal linear time in the acyclic case, and in $O(n\log n)$ time in the general case. 
    \item[(iv)] We show how to compute the smallest WDFA equivalent to \emph{any acyclic DFA} in nearly-optimal time.
\end{enumerate}
Our contributions imply new results of independent interest.
Contributions (i-iii) extend the universe of known regular languages for which membership can be tested efficiently [Backurs and Indyk, FOCS 2016] and provide a new class of NFAs for which the minimization problem can be approximated within constant factor in polynomial time. In general, the NFA minimization problem does not admit a polynomial-time $o(n)$-approximation unless P=PSPACE.
Contribution (iv) is a big step towards a complete solution to the well-studied problem of indexing graphs for linear-time pattern matching queries: our algorithm provides a provably minimum-size solution for the deterministic-acyclic case.
\end{abstract}

\newpage

\clearpage
\setcounter{page}{1}

\section{Introduction}

Prefix-sorting is the process of ordering the positions of a string in the co-lexicographic order of their corresponding prefixes\footnote{Usually, the lexicographic order is used to sort string suffixes. In this paper, we use the symmetric co-lexicographic order of the string's prefixes, and extend the concept to labeled graphs.}. Once this step has been performed, several problems on strings become much easier to solve: for example, substrings can be located efficiently in the string without the need to read all of its characters. Given the versatility of this tool, it is natural trying to generalize it to more complex objects such as edge-labeled trees and graphs. For example, a procedure for lexicographically-sorting the states of a finite-state automaton could be useful to speed up subsequent membership queries in its accepting language or its substring/suffix closure; as shown by Backurs and Indyk~\cite{backurs2016regular}, membership and pattern matching problems on regular languages are hard in the general case. 
The newborn theory of \emph{Wheeler graphs}~\cite{gagie2017wheeler} provides such a generalization.
Intuitively, a labeled graph is Wheeler if and only if its nodes can be co-lexicographically sorted in a total order, i.e. pairwise-distinct nodes are ordered according to (i) their incoming labels or (when the labels are equal), according to (ii) their predecessors.
As a consequence, Wheeler graphs admit indexes for linear-time exact pattern matching queries (also known as \emph{path queries}).
Wheeler graphs generalize several lexicographically-sorted structures studied throughout the past decades: indexes on strings~\cite{manber1993suffix,weiner1973linear,ferragina2005indexing}, sets of strings~\cite{mantaci2007extension}, trees~\cite{ferragina2009compressing}, de Bruijn graphs~\cite{de1946combinatorial}, variable-order de Bruijn graphs~\cite{siren2014indexing}, wavelet trees~\cite{grossi2003high}, wavelet matrices~\cite{claude2015wavelet}. 
These efforts are part of a more general wave of interest (dating as far back as 27 years ago~\cite{manber1992approximate}) towards techniques aimed at solving pattern matching on labeled graphs~\cite{equi2019complexity,siren2014indexing,gagie2017wheeler,ferragina2009compressing,jain2019complexity,vaddadi2017sequence,rautiainen2017aligning,amir2000pattern,equi2019complexity2,manber1992approximate}. 
%Wheeler graphs generalize those techniques that exploit indexing to solve the problem. 
As discussed above, existing graph-indexing solutions can only deal with simple labeled graphs. The problem of indexing general (or even just acyclic) graphs with a solution of provably-minimum size remains unsolved. 
Unfortunately, not all graphs are Wheeler and, as Gibney and Thankachan~\cite{gibney2019hardness} have recently shown, the problem of recognizing (and sorting) them turns out to be NP-complete even when the graph is acyclic (this includes, in particular, acyclic NFAs). Even worse, not all regular languages admit an accepting Wheeler finite automaton: the set of \emph{Wheeler languages} is a proper superset of the finite languages and a proper subset of the regular languages~\cite{gagie2017wheeler}. 
Even when an index is not used, exact pattern matching on graphs is hard: Equi et al.~\cite{equi2019complexity,equi2019complexity2} have recently shown that any solution to the problem requires at least quadratic time (under the Orthogonal Vectors hypothesis), even on acyclic DFAs. In particular, this implies that converting an acyclic DFA into an equivalent Wheeler DFA cannot be done in less than quadratic time in the worst case. 

%In this paper, we initiate the study of Wheeler languages, that is, regular languages admitting an accepting automaton that can be prefix-sorted. We moreover tackle the questions: 

The remaining open questions, therefore, are: what are the properties of Wheeler languages? which class of automata admits polynomial-time prefix-sorting procedures? Can we efficiently build the smallest (that is, with the minimum number of states) prefix-sortable finite-state automaton that accepts a given regular language? 
%Note that, in order to answer the latter question, it makes sense to first focus on finite languages (which always admit a solution) and on DFAs, considering the hardness of the NFA-minimization~\cite{gruber2007computational,Malcher:2004:MFA:1046130.1046139} and Wheeler NFAs recognition~\cite{gibney2019hardness} problems (both hard also in the acyclic case).
These questions are also of practical relevance: as shown in~\cite{siren2014indexing}, acyclic DFAs recognizing \emph{pan-genomes} (i.e. known variations in the reference genome of a population) can be turned into equivalent WDFAs of the same expected asymptotic size. While the authors do not find the minimum such automaton, their theoretical analysis (as well as experimental evaluation) suggests that the graph-indexing problem is tractable in some real-case scenarios.

\subsection{Our Contributions}

In this paper we provide the following contributions: 

\begin{enumerate}
    \item We show that Wheeler languages are the natural version of regular languages endowed with the co-lexicographic ordering: 
    %when sorted, the strings belonging to a Wheeler language are partitioned in a \emph{finite} number of Myhill-Nerode equivalent \emph{intervals}. 
    when sorted, the strings belonging to a Wheeler language are partitioned into a \emph{finite} number of  \emph{intervals}, each formed by elements from a single Myhill-Nerode equivalence class. In regular languages, those \emph{intervals} are replaced with general \emph{sets}.
    %\emph{finite} number of Myhill-Nerode equivalent \emph{sets}.
    \item We show that every Wheeler NFA (WNFA) with $n$ states admits an equivalent Wheeler DFA (WDFA) with at most $2n-1-|\Sigma|$ states ($\Sigma$ being the alphabet) that can be computed in $O(n^3)$ time. This is in sharp contrast with general NFAs
    (where the blow-up could be exponential).
	\item Let $d$-NFA denote the class of NFAs with at most $d$ equally-labeled edges leaving any state. We show that the problem of recognizing and sorting Wheeler $d$-NFAs is in P for $d\leq 2$.
	A recent result from Gibney and Thankachan~\cite{gibney2019hardness} shows that the problem is NP-complete for $d\geq 5$. Our result almost completes the picture, the remaining open cases being $d=3$ and $d=4$. 
	\item We provide an online incremental algorithm that, when fed with an acyclic Wheeler DFA's nodes in any topological order, can dynamically compute the co-lexicographic rank of each new incoming node among those already processed with just logarithmic delay.
	\item We improve the running time of (4) to linear in the offline setting for arbitrary WDFAs.
	%We show that the problem of recognizing and sorting Wheeler DFAs can be solved in optimal linear time with an offline algorithm. 
	\item Given a Wheeler DFA $\mathcal A$ of size $n$,  we show how to compute, in $O(n\log n)$ time, the smallest Wheeler DFA recognizing the same language as $\mathcal A$. If $\mathcal A$ is acyclic, running time drops to $O(n)$.
	\item Given \emph{any acyclic DFA} $\mathcal A$ of size $n$, we show how to compute, in $O(n+m\log m)$ time, the smallest Wheeler DFA $\mathcal A'$, of size $m$, recognizing the same language as $\mathcal A$.
\end{enumerate}

The paper is structured in a top-down fashion to make it accessible also to non-experts. We provide the main theorems and proof sketches within the first ten pages. 
%All technical proofs can be found in the appendix. 
We start in Section \ref{sec:WL} with results {\bf 1} and {\bf 2}: a Myhill-Nerode theorem for Wheeler languages and a linear conversion from WNFAs to WDFAs. Result {\bf 1} shows that Wheeler languages are precisely those admitting a ``finite interplay'' between the co-lexicographic ordering and the Myhill-Nerode equivalence relation (i.e. the relation characterizing general DFAs). Result {\bf 2} implies that the WNFA minimization problem admits a polynomial-time  $2$-approximation. We remark that the NFA minimization problem is notoriously hard (even to approximate within $o(n)$-factor) in the general case~\cite{approx_min_NFA,gruber2007computational,Malcher:2004:MFA:1046130.1046139}. 
In Section \ref{sec:sorting} we describe results {\bf 3-5}: polynomial-time algorithms for recognizing and sorting Wheeler $2$-NFAs. %The algorithm is based on a reduction to 2-SAT; we expose it first since in our opinion it is useful to get an intuitive overview of the problem.
%
%Our algorithms are described using intuitive high-level operations (e.g. manipulation of dynamic sequences and spanning tree computation) which do not require prior data-structures knowledge. All technical data-structures details showing how to implement those operations are given separately in Appendix \ref{sec:data structures}. 
These results generalize to labeled graphs existing prefix-sorting algorithms on strings~\cite{navarro2014optimal} and labeled trees~\cite{ferragina2009compressing} that have been previously described in the literature.
Combined with contribution {\bf 2}, our algorithms can be used to index any Wheeler NFA at the price of a moderate linear blow-up in the number of states. This result expands the universe of known regular languages for which membership and pattern matching problems can be solved efficiently~\cite{backurs2016regular}.
Contributions {\bf 6} and {\bf 7} (Section \ref{sec:minimization}) combine our sorting algorithms {\bf 3-5} with DFA minimization techniques to solve the following problem: to compute, given a finite language represented either explicitly by a set of strings or implicitly by an acyclic DFA, the smallest accepting WDFA. 
%that can be prefix-sorted and indexed for efficient path queries. 
%Running time is optimal (linear) if the language is provided as a set of strings (using contribution  {\bf 4}), or a logarithmic multiplicative factor from the optimal if the input is an acyclic DFA (using contribution  {\bf 5}).
Note that this can be interpreted as a technique to index arbitrary deterministic acyclic graphs using the smallest prefix-sortable equivalent automaton.
While we do not provide a lower bound stating that the \emph{space} of our index is minimum, we note that all known fast indexes on strings, sets of strings, trees, and variable-order de Bruijn graphs are Wheeler graphs~\cite{gagie2017wheeler}. In this sense, any index on acyclic graphs improving our solution would probably require techniques radically different than those developed in the last decades to solve the indexing problem (virtually any known full-text index uses prefix sorting, including those based on LZ77~\cite{KN13}, run-length BWT~\cite{gagie2018optimal}, and grammars~\cite{CNspire12}).
%This section requires knowledge of basic automata-theory concepts such as the Myhill-Nerode equivalence and DFA minimization.
%Finally, in Appendix \ref{sec:indexing} we show how to index languages recognized by Wheeler finite automata using simple data structures. 

\subsection{Definitions}\label{sec:definitions}

We start by giving a definition of finite-state automata that captures, to some extent, the amount of nondeterminism of the automaton. A $d$-NFA is a nondeterministic finite state automaton that has at most $d$ transitions with the same label leaving each state. Note that $1$-NFAs correspond to DFAs, while $\infty$-NFAs correspond to NFAs. 
%Our definition is graph-theoretic: rather than defining a transition function, we treat automata as edge-labeled directed graphs with a source and a set of accepting states. This will be useful later when extending the notion of \emph{prefix sorting} to labeled graphs.

\begin{definition}\label{def:d-NFA}
	A $d$-NFA is a quintuple $\mathcal A = (V,E,F,s,\Sigma)$, where $V$ is a set of states (or vertices), $\Sigma$ is the alphabet (or set of labels), $E\subseteq V\times V \times \Sigma$ is a set of directed labeled edges, $F\subseteq V$ is a set of accepting states,  and $s \in V$ is a start state (or source).
	%\begin{itemize}
	%	\item $V$ is a set of states (or vertices),
	%	\item $\Sigma$ is the alphabet (or set of labels),
	%	\item $E\subseteq V\times V \times \Sigma$ is a set of directed labeled edges,
	%	\item $F\subseteq V$ is a set of accepting states,  and
	%	\item $s \in V$ is a start state (or source).
	%\end{itemize}
	We moreover require that $s\in V$ is the only node with in-degree zero and that for each $u\in V$ and $a\in\Sigma$, $|\{(u,v,a)\in E\}| \leq d$. 
\end{definition}

We denote $\sigma = |\Sigma|$.
The notation $\mathcal L(\mathcal A)$ indicates the language accepted by $\mathcal A$, i.e. the set of all strings labeling paths from $s$ to an accepting state. 
We assume that each state either is $s$ or is reachable from $s$. Otherwise, any state that cannot be reached from $s$ can be removed without changing $\mathcal L(\mathcal A)$. 
Note that we allow states with incomplete transition function, i.e. such that the set of labels of their outgoing edges does not coincide with $\Sigma$. If state $s$ misses outgoing label $a$, then any computation following label $a$ from $s$ is considered as non-accepting. In a standard NFA definition, this would be equivalent to having an outgoing edge labeled $a$ to a universal non-accepting node (a sink).
%We note that a NFA with multiple sources can be reduced to the above definition by collapsing all sources in one source (in this case, $d$ becomes the maximum between the nondeterminism degree and the number of sources of the original NFA). 
We call a $d$-NFA \emph{acyclic} when the graph $(V,E)$ does not have cycles. We say that a $d$-NFA is \emph{input-consistent} if, for every $v\in V$, all incoming edges of $v$ have the same label. If the $d$-NFA is input-consistent, we indicate with $\lambda(v)$, $v\in V$, the label of the incoming edges of $v$. For the source, we take $\lambda(s) = \# \notin \Sigma$.
We will assume that characters in $\Sigma$ are totally ordered by $\prec$  and that $\#$ is minimum. We extend $\prec$ to $\Sigma^*$ co-lexicographically,  still denoting it  by $\prec$. 
On DFAs, we denote by $succ_a(u)$, with $u\in V$ and $\ a\in\Sigma$, the unique successor of $v$ by label $a$, when it exists. 
We define the \emph{size} of an automaton to be the number of its edges. 
The notion of \emph{Wheeler graph} generalizes in a natural way the concept of co-lexicographic sorting to labeled graphs:

\begin{definition}[Wheeler Graph]\label{def_WG} % Slightly modified copy-paste from the tunneling on Wheeler graphs paper
	A triple $G=(V, E, \Sigma)$, where $V$ is a set of vertices and $E \subseteq V \times V \times \Sigma$ is a set of labeled edges, is called a \emph{Wheeler graph} if there is a total ordering $<$ on $V$ such that vertices with in-degree $0$ precede those with positive in-degree and for any two edges $(u_1,v_1,a_1), (u_2, v_2,a_2)$ we have
	(i) $a_1 \prec a_2 \rightarrow v_1 < v_2$, and (ii) $\left(a_1=a_2 \right) \wedge \left( u_1 <  u_2 \right) \rightarrow v_1 \leq v_2$.
	%\begin{itemize}
	%	\item[(i)] $a_1 \prec a_2 \rightarrow v_1 < v_2$,
	%	\item[(ii)] $\left(a_1=a_2 \right) \wedge \left( u_1 <  u_2 \right) \rightarrow v_1 \leq v_2$,
	%\end{itemize}
	%where $\prec$ denotes the ordering on $\Sigma$.
\end{definition}

Note that the above definition generalizes naturally the concept of prefix-sorting from strings to graphs: two nodes (resp. string prefixes) can be ordered either looking at their incoming labels (resp. last characters) or, if the labels are equal, by looking at their predecessors (resp. previous prefixes).
Considering that, differently from strings and trees, a graph's node can have multiple predecessors, it should be clear that there could exist graphs whose nodes cannot be sorted due to conflicting predecessors: not all labeled graphs enjoy the Wheeler properties. 
We call a total order of nodes satisfying Definition \ref{def_WG} a \emph{Wheeler order} of the nodes and we write \emph{WDFA,WNFA} as a shortcut for \emph{Wheeler DFA, Wheeler NFA}.
By property (i), input-consistence is a necessary condition for a graph to be Wheeler. 
%We note that an input-consistent graph is, essentially, a state-labeled graph: since each state has only one distinct incoming label, we can move the label from the incoming edges to the state itself. While we believe this interpretation is more intuitive, we preferred to stick to the original definition given in~\cite{gagie2017wheeler} and consider edge-labeled graphs instead.
%We note this could be relaxed by replacing $<$ with $\leq$ in property (i); in this case however, it can be seen that at most $\sigma-1$ nodes can have distinct incoming labels. For simplicity, in this paper we stick to the original definition. 
%We note that input-consistence is not a restrictive condition: given an automaton that is not input-consistent, we can expand each node $v$ to nodes $v_1, \dots, v_\sigma$, where each $v_c$ keeps only the incoming edges of $v$ labeled $c$ and all the outgoing edges of $v$. This increases the automaton's size by at most a $\sigma$-multiplicative factor. % Note from Jarno: we can also manage with only a factor 2 increase by splitting edges kind-of like in my reduction from WG sorting to WG recognition. This changes the language of the graph, but we can still match against the original language by taking into account this splitting when matching.
%Nicola: right! I commented out that part in an effort to fit everything in the first 10 pages :)
An important property of Wheeler graphs is \emph{path coherence}:

\begin{definition}[Path coherence \cite{gagie2017wheeler}]\label{def:path coherence} An edge-labeled directed graph $G$ is path coherent if there is a total order of the nodes such that for any consecutive range $[i,j]$ of nodes and string $\alpha$, the nodes reachable from those in $[i,j]$ in $|\alpha|$ steps by following edges whose labels form $\alpha$ when concatenated, themselves form a consecutive range.
\end{definition}
A Wheeler graph is path coherent with respect to any Wheeler order of the nodes \cite{gagie2017wheeler}. 
%We will show that Wheeler finite-state automata have an intuitive interpretation: each state $u$ is associated with a (possibly infinite) set $S_u$ of strings, i.e. those labeling paths connecting the source with $u$. The automaton is Wheeler if and only if there exists a total ordering $<$ of the states such that $u< v$ if and only if all strings in $S_u$ are co-lexicographically smaller than or equal to those in $S_v$.

%As said above, Wheeler graphs generalize many data structures on strings, trees, and graphs exploiting suffix sorting to speed up queries. Given a prefix-sorted Wheeler graph, Gagie et at.~\cite{gagie2017wheeler} show how to build a small data structure supporting \emph{path queries} on the graph: in our case (finite automata), this means to find, given a string $\alpha$, the range of nodes $u$ such that $\alpha\in S_u$. 
%In our case (finite automata), this can be used to find, given a string $\alpha$, all nodes $u$ such that $\alpha\in S_u$. The Wheeler properties imply that all such nodes form a contiguous range in the ordering of the nodes. 
%This will allow us to quickly recognize strings belonging to the language of particular Wheeler NFAs (a notoriously hard problem on general NFAs when both pre-processing and query times are required to be small).

\section{Wheeler Languages}\label{sec:WL}

In this section we collect our basic results on regular languages accepted by automata whose transition relation is a Wheeler graph: Wheeler languages.
\begin{definition} Let $\Sigma$ be a finite set. A language $\mathcal L$ is {\em Wheeler} if $\mathcal L= \mathcal L(\mathcal A)$ for a Wheeler NFA $\mathcal A$. 
 \end{definition}
 Let us begin with some basic notation. Given a language $\mathcal L \subseteq \Sigma^*$ we denote by $ \text{\em Pref}(\mathcal L),  \text{\em Suff}(\mathcal L), $ and $ \text{\em Fact}(\mathcal L)$ the set of prefixes, suffixes, and factors  of strings in  $\mathcal L$, respectively. More formally:
 $ \text{\em Pref}(\mathcal L)   = \{\alpha: \exists \beta \in \Sigma^*~ \alpha \beta\in \mathcal L\}$, $\text{\em Suff}(\mathcal L)   = \{\beta: \exists \alpha \in \Sigma^*~ \alpha \beta\in \mathcal L\}$,  $\text{\em  Fact}(\mathcal L)   = \{\alpha: \exists \beta, \gamma \in \Sigma^*~ \gamma \alpha \beta\in \mathcal L\}$. Given two states $u,v$ of an NFA $\mathcal{A}$, we denote by $u \rightsquigarrow v$ a path from $u$ to $v$ in $\mathcal{A}$.

 \begin{definition}
 If $\mathcal A=(V,E,F,s,\Sigma)$ is an NFA, $ u\in V $, and $\alpha\in  \text{\em  Pref}(\mathcal L(\mathcal A))$, we define:
 \begin{enumerate}
\item  $V_\alpha=\{ v ~|~ \alpha \text{ labels } s \rightsquigarrow v \}$,
\item $ \text{\em  Pref}(\mathcal L(\mathcal A))_u:=\{\alpha \in  \text{\em  Pref}(\mathcal L(\mathcal A)):  \alpha \text{ labels } s \rightsquigarrow u\} $ .
\end{enumerate}
 \end{definition}
Clearly, from the above definition it follows that (i) $V_{\alpha} \subseteq V $,  (ii) $  \text{\em  Pref}(\mathcal L(\mathcal A))_{u}\subseteq \text{\em Pref}(\mathcal L(\mathcal A)) $, and (iii) 
$u \in V_{\alpha} \text{ if and only if } \alpha \in  \text{\em  Pref}(\mathcal L(\mathcal A))_{u}$.
%\begin{align*}
% u \in V_{\alpha} & \text{ if and only if } \alpha \in  \text{\em  Pref}(\mathcal L(\mathcal A))_{u}.
% \end{align*}

The prefix/suffix property introduced below is going to be crucial to determine the Wheeler ordering among states---when such an ordering exists.
 \begin{definition}  
 Consider a linear order $(L, <)$.
 \begin{enumerate}
 \item An {\em interval} in  $(L, <)$  is a  $I\subseteq L$  such that
 $(\forall x,x'\in I)(\forall y \in L) (x< y< x' \rightarrow y\in I)$.
 \item Given $I,J$   intervals in  $(L, <)$ and  $I\subseteq J$, then: 
 \begin{itemize}
\item[-] $I$ is a {\em prefix} of $J$ if $(\forall x \in I )( \forall y \in J\setminus I) ( x< y)$;
\item[-]  $I$   is a    {\em suffix} of $J$ if   $(\forall y \in J \setminus I )( \forall x \in I )( y< x)$.
\end{itemize}
\item A family $\mathcal C$ of non-empty intervals in  $(L, <)$ is said to have the  {\em prefix/suffix property} if, for all $I,J \in \mathcal C$ such that  $I \subseteq J$, $I$ is either a prefix or a suffix  of $J$.
 \end{enumerate}
  \end{definition}
  
The following lemma will allow us to bound (linearly) the blow-up of the number of states taking place when moving from a WNFA to a WDFA. 
 
\begin{lemma}\label{3n}
Let $(L, <)$ be a finite linear order of cardinality $|L|=n$, let  $\mathcal C$ be a {\em prefix/suffix} family of non-empty intervals in $(L, <)$. Then:
\begin{enumerate}
    \item $|\mathcal C| \leq 2n-1$.
    \item The upper bound is tight: for every $n$, there exists a prefix/suffix family of size $2n-1$.
\end{enumerate}
\end{lemma}

 \begin{definition}
 Let $\mathcal C$ be a family of non-empty intervals of a linear order  $(L, <)$ having the prefix/suffix property. Let $<^{i}$ (or simply $ < $) the binary relation over $\mathcal C$ defined by
 \begin{align*}
 I<^{i} J & \text{ if and only if } ( \exists x \in I)( \forall y\in J )( x< y )\lor (\exists y \in J)( \forall x \in I )( x< y).
 \end{align*}
 \end{definition}
The following lemma is easily proved.
 
 \begin{lemma} \label{convex+order}  
 Let $\mathcal C$ be a family of non-empty intervals of a linear order  $(L, <)$ having the prefix/suffix property, then
 $ (\mathcal C,<^{i}) $ is a linear order.  
 \end{lemma}

 Note that whenever the linear order   $(L, <)$ is finite,  any non-empty interval $I$  has minimum $m_I$  and maximum $M_I$. In this special case, the above order  $<^{i}$ can be equivalently described on a family having the prefix/suffix property, by: 
$
I  <^{i} J  \text{ if and only if } (m_I <  m_J) \lor [ (m_I=m_J) \wedge (M_I < M_J)].$

\medskip

We now have the basics to start our study of Wheeler languages.  In this section, we use $V,E,F,s,\Sigma,<$  to denote the set of states, edges, final states, initial state, alphabet, and Wheeler order of  a generic WNFA. The key property of path-coherence will be re-proved below---in Lemma \ref{convex_sets}---, together with what we may call a  sort of its  ``dual'', that is,  the the set of strings read while reaching a given state is an interval. More precisely, if  $\mathcal A$ is a WNFA, $ u\in V $, and $\alpha\in \Sigma^*$, we have that $ V_{\alpha}$ is an interval in $ (V,<) $ ($ I_{\alpha} $, from now on), $ \text{\em  Pref}(\mathcal L(\mathcal A))_{u}$ is an interval in $(\text{\em  Pref}(\mathcal L(\mathcal A)),\prec)$ ($ I_{u} $, from now on), and
\[ \alpha \in I_u  \text{ if and only if } u \in I_\alpha.\]

Preliminary to our result are the following lemmas, exploiting the interval-structure of both Wheeler languages and automata.
 \begin{lemma}\label{prec_versus_minus} 
 If $\mathcal A$  is a   WNFA,   $u,v\in V$   are states, and $\alpha, \beta \in \text{Pref}(\mathcal L(\mathcal A))$, then:
 \begin{enumerate}
\item  if  $\alpha\in I_u, \beta\in I_v$, and  $\{\alpha,\beta\}\not\subseteq I_v\cap I_u$, 
 then $\alpha\prec \beta  \text{ if and only if } u<v$;
\item 
 if  $u\in I_\alpha, v\in I_\beta $, and  $\{u,v\}\not\subseteq I_\beta \cap I_\alpha$,
 then $\alpha\prec \beta \text{ if and only if } u<v$.
 \end{enumerate}
 \end{lemma}

Let $ I_{V} =\{I_u:~ u\in V\}$ and  $I_ {\text{\em Pref}(\mathcal L(\mathcal A))}  =\{I_\alpha :~\alpha\in \text{\em Pref}(\mathcal L(\mathcal A)) \}$.

  \begin{lemma}\label{convex_sets}  
 If $\mathcal A$  is a  WNFA  and $\mathcal L = \mathcal L(\mathcal A)$, then:
 \begin{enumerate}
 \item  for all  $u\in V$,  the set  $I_u$  is  an interval  in $( \text{Pref}(\mathcal L(\mathcal A)), \prec)$; 
 \item $ I_{V}$ is  a prefix/suffix family  of intervals   in  $( \text{Pref}(\mathcal L(\mathcal A)), \prec)$;
 \item  for all  $\alpha \in \text{Pref}(\mathcal L(\mathcal A))$,  the set  $I_\alpha$    is  an interval  in $(V,<)$; 
 \item $I_ {\text{Pref}(\mathcal L(\mathcal A))}$  is a prefix/suffix family  of intervals in $(V,<)$.
 \end{enumerate}
 \end{lemma}

From  Lemma \ref{convex+order} it  follows that   both $(I_{v}, \prec^{i})$  and  $(I_ {\text{\em Pref}(\mathcal L(\mathcal A))}, <^{i})$  are linear orders.   The link between such orders is made explicit below.

\begin{lemma}\label{compare} Consider $ I_{u},I_{v}\in I_{V} $ and $ I_{\alpha}, I_{\beta} \in I_ {\text{Pref}(\mathcal L(\mathcal A))}$.  
\begin{enumerate}
\item $ I_u  \prec^{i} I_v $ implies that $ u < v $ and  $ u < v $ implies that  $ I_u  \preceq^{i} I_v $
\item $ I_{\alpha} <^{i} I_{\beta} $ implies that $ \alpha \prec  \beta $ and  $ \alpha \prec  \beta $ implies that  $ I_{\alpha} \leq^{i} I_{\beta} $
\end{enumerate}
\end{lemma}

If $\mathcal A$  is a   WNFA we can prove that the following interval construction---which is the analogous of the power-set construction for NFAs---allows determinization.
\begin{definition}
 If $\mathcal A$  is a   WNFA  we define  its (Wheeler) \emph{determinization} as the automaton ${\mathcal A^{d}}=(V^{d}, E^{d},F^{d} s^{d},<^{d}, \Sigma)$, where:
 \begin{itemize}
     \item[-] $V^{d}=I_{\text{ Pref}(\mathcal L(\mathcal A))}$;
     \item[-] $s^{d}= I_{\epsilon}=\{s\}$
     \item[-] $F^{d}=\{I_\alpha ~|~ \alpha \in \mathcal L(\mathcal A)\}$;
     \item[-] $E^{d}$ is the set of triples $(I_\alpha,I_{\alpha e}, e)$, for all  $e\in \Sigma$ and $\alpha e\in \text{Pref}(\mathcal L(\mathcal A))$;
   \item[-] $<^{d}=<^{i}$. 
 \end{itemize} 
\end{definition}

The bound proved in Lemma \ref{3n} can be sligthly improved in the special case of prefix/suffix families corresponding to WNFA intervals.
\begin{lemma}[Determinization of a Wheeler NFA]\label{Wdeterminization}
 If $\mathcal A$  is a   WNFA with $n$ states over an alphabet $\Sigma$,  then $\mathcal A^{d}$ is a WDFA  with at most $2n-1-|\Sigma|$ states, and  $ \mathcal L(\mathcal A^{d}) = \mathcal L(\mathcal A)$. 
\end{lemma}

 \begin{lemma}[Computing the determinization]\label{ComputeWdeterminization}
 If $\mathcal A$  is a   WNFA with $n$ states,  then $\mathcal A^{d}$ can be computed in $O(n^3)$ time.
\end{lemma}
  
Lemma \ref{Wdeterminization} above--- saying that we can restrict the automata recognizing Wheeler Languages to deterministic ones without an exponential blow up---marks a difference between the standard and the Wheeler case for regular languages and can be seen as the first step in the study of Wheeler Languages. Further differences can be observed. For example, the reader can check that the language $ \mathcal L(\mathcal A)=\mathcal L= b^{+}a $ is accepted by \emph{incomplete} WDFAs only. 

The subsequent step to take in a theory of Wheeler Languages is a Myhill-Nerode like theorem for this class.  To this end, we define:

 \begin{definition}  \label{right_inv} Given a language $\mathcal L\subseteq \Sigma^*$, an equivalence relation $\sim$ over $\text{\em Pref}(\mathcal L)$    is:
\begin{itemize}
\item[-]   {\em right invariant}, when for all  $\alpha, \beta \in \text{\em Pref}(\mathcal L)$ and $\gamma\in \Sigma^*$: 
 if $ \alpha \sim \beta$ and  $\alpha\gamma \in \text{\em Pref}(\mathcal L)$, then   $\beta\gamma \in \text{\em Pref}(\mathcal L)$ and $ \alpha\gamma \sim \beta\gamma$;
\item[-] {\em convex} if   $\sim$-classes are intervals of $(\text{\em Pref}(\mathcal L),\prec)$;
\item[-]   {\em input consistent} if all  words belonging to the same $\sim$-class end with the same letter. 
\end{itemize}
\end{definition}

Consider a  language $\mathcal L\subseteq \Sigma^*$. The Myhill-Nerode equivalence  $\equiv_{\mathcal L}$ among words in $  Pref(L)$  is defined as   
\begin{align*}
\alpha \equiv_{\mathcal L} \beta \text{ if and only if } & (\forall \gamma \in \Sigma^{*}) ( \alpha \gamma \in {\mathcal L} \Leftrightarrow \beta \gamma \in {\mathcal L} ). 
\end{align*}
\begin{definition}
The input consistent, convex refinement of $ \equiv_{\mathcal L} $ is defined as follows. 

For all $\alpha, \beta \in  \text{Pref}(\mathcal L)$:
\begin{align*}
\alpha\equiv_{\mathcal L}^{c}\beta  \text{ if and only if } &  \alpha \equiv_{\mathcal L} \beta \wedge ~end(\alpha)=end(\beta)  \wedge (\forall    \gamma \in  \text{\em  Pref}(\mathcal L))  (\alpha\prec \gamma \prec \beta \rightarrow \gamma \equiv_{\mathcal L} \alpha),
\end{align*}
where $end(\alpha)$ is the final character of $\alpha$ when $\alpha\neq\epsilon$, and $ \epsilon $ otherwise. 
\end{definition}
Using  the above results in this section we can prove:

\begin{theorem}[Myhill-Nerode for Wheeler Languages]  \label{Myhill-Nerode} 
Given a language $\mathcal L\subseteq \Sigma^*$, the  following are equivalent:
\begin{enumerate}
\item $\mathcal L$ is a Wheeler language (i.e. $\mathcal L$  is recognized by a  WNFA).
\item $\equiv_{\mathcal L}^{c}$ has  finite index.
\item $\mathcal L$ is a union of   classes of a convex, input consistent, right invariant  equivalence over  $\text{Pref}(\mathcal L)$  of finite index.
\item $\mathcal L$ is recognized by a  WDFA.
\end{enumerate}
\end{theorem}

This theorem and other results on Wheeler Languages are going to be part  of a companion paper of this one.

\section{Sorting Wheeler Finite Automata}\label{sec:sorting}

In this section we provide efficient algorithms to sort a relevant sub-class of the Wheeler automata. 
Combined with the results of the previous section, our algorithms can be used to index \emph{any} WNFA.
%(after converting it to an equivalent WDFA).
%\subsection{Recognizing Wheeler $2$-NFAs is in P}\label{sec:NFA}
We start with a reduction from the problem of recognizing Wheeler $2$-NFAs to 2-SAT. The reduction introduces only a polynomial number of boolean variables and can be computed in polynomial time; since 2-SAT is in P, this implies that Wheeler $2$-NFA recognition is in P. 
%Our reduction actually achieves more than this: if the graph is a Wheeler $2$-NFA, then the truth assignment of the formula's variables yields a total node ordering satisfying the Wheeler property. 
%In the next sections we discuss faster algorithms for DFAs. 
%We decided to present first our reduction to 2-SAT since we believe it gives a clearer understanding of the problem. For space constraints, here we only give an overview of the proof. 

\begin{theorem}\label{thm:2NFA in P}
	Let  $\mathcal A=(V,E,F,s,\Sigma)$ be a $2$-NFA. In $O(|E|^2)$ time we can:
	\begin{enumerate}
		\item Decide whether $\mathcal A$ is a Wheeler graph, and
		\item If $\mathcal A$ is a Wheeler graph, return a node ordering satisfying the Wheeler graph definition. 
	\end{enumerate}
\end{theorem}
\begin{proof}[Sketch] 
It is easy to express the Wheeler properties (i)-(ii) and antisymmetry/connex of the Wheeler order with 2-SAT clauses. Transitivity, however, requires 3-SAT clauses on general graphs. The core of the full proof in Appendix \ref{proof thm 2NFA in P} is to show that, on $2$-NFAs, transitivity automatically ``propagates'' from the source to all nodes and does not require additional clauses.\qed
\end{proof}

Gibney and Thankachan~\cite{gibney2019hardness} have recently shown that the problem of recognizing Wheeler $d$-NFAs is NP-complete for $d\geq 5$. Theorem \ref{thm:2NFA in P} almost completes the picture, the remaining open cases being $d=3$ and $d=4$.  We note that Theorem \ref{thm:2NFA in P} combined with our determinization result of Section \ref{sec:WL} does not break the problem's NP-completeness: in principle, our determinization algorithm could turn a non-Wheeler NFA into a WDFA. 
%This is due to the fact that, in order to preserve the Wheeler property, the topology of those trees cannot be arbitrary and must satisfy the co-lexicographic ordering of the nodes, i.e. the solution we are trying to compute. While there always exists a choice of tree topologies preserving the Wheeler property, computing the reduction turns out to be as complicated as finding a solution to the original problem.

%\subsection{Sorting Acyclic Wheeler DFAs Online}\label{sec:DFA}

We now describe more efficient algorithms for the deterministic case. 
The first, Theorem \ref{thm:n log n}, is an online algorithm that solves the problem considered in Theorem \ref{thm:2NFA in P} in $O(|E|\log|V|)$ time when the graph is an acyclic DFA.  The algorithm is online in the following sense. We assume that the nodes, together with their incoming labeled edges, are provided to the algorithm in any valid topological ordering.
At any step, we maintain a prefix-sorted list of the current nodes, which is updated when a new node is added.
When a new node $v$ arrives together with its incoming labeled edges $(u_1,v,a), \dots, (u_k,v,a)$, then $u_1,\dots, u_k$ have already been seen in the past node sequence and can be used to decide the co-lexicographic rank of $v$. If $v$ falsifies the Wheeler properties, we detect this event, report it, and stop the computation. 
Our algorithm is an extension of an existing one that builds online the Burrows-Wheeler transform of a string~\cite{navarro2014optimal}.
In Section \ref{sec:minimization} we will modify this algorithm so that, instead of failing on non-Wheeler graphs, it computes the smallest Wheeler DFA equivalent to the input acyclic DFA.

First, note that 
Lemma \ref{prec_versus_minus} implies that Wheeler DFAs admit a unique admissible ordering (this follows from the fact that, on WDFAs, $\{\alpha,\beta\}\not\subseteq I_v\cap I_u$ always holds):

\begin{corollary}\label{lem:clusters}
	Let $\mathcal A$ be a Wheeler DFA, $<$ be the node ordering satisfying the Wheeler properties, and $\prec$ be the co-lexicographic order among strings. For any two nodes $u \neq v$, the following holds: $\alpha_u \prec \alpha_v$ for all string pairs $\alpha_u, \alpha_v$ labeling paths $s \rightsquigarrow u$ and $s\rightsquigarrow v$ if and only if $u<v$. 
\end{corollary}
%\begin{proof}[Sketch]
%	Consider two strings $\alpha_u$ and $\alpha_v$ labeling two paths from the source to $u$ and $v$. Since the automaton is deterministic and $u \neq v$, it must be the case that $\alpha_u \neq \alpha_v$. At this point, the proof works by induction on the length of $\alpha_u$ and $\alpha_v$, using the Wheeler properties at the base case where $\alpha_u$ and $\alpha_v$ end with two distinct characters (or one of the two nodes is the source).
%	\qed
%\end{proof}

Corollary \ref{lem:clusters} has two important consequences: on DFAs, (i) we can use \emph{any} paths connecting $s$ with two nodes $u\neq v$ to decide their co-lexicographic order, and (ii) if it exists, the total ordering of the nodes is unique. The corollary is crucial in proving the following (as well as others) result:

\begin{theorem}\label{thm:n log n}
	Let  $\mathcal A=(V,E,F,s,\Sigma)$ be an acyclic DFA. There exists an algorithm that either prefix-sorts the nodes of $\mathcal A$ or returns $\mathtt{FAIL}$ if such an ordering does not exist online with $O(\log|V|)$ delay per input edge.
\end{theorem}

All details of our algorithm (description, pseudocode and data structures) and the proof of its correctness are reported in in Appendices \ref{app: Sorting WDAGs Online} and \ref{sec:data structures}.

%\subsection{Sorting Wheeler DFAs in Linear Time}\label{sec:DFA linear}

To conclude the section, we show that in the offline setting we can improve upon the previous result.
%: Wheeler DFAs can be recognized and sorted (offline) in linear time. 
We first need the following lemma (see Appendix \ref{proof lemma check range consistency} for the full proof):

\begin{lemma}\label{lem: check range consistency}
	Given an input-consistent edge-labeled graph $G=(V,E,\Sigma)$ and a permutation of $V$ sorted by a total order $<$ on $V$, we can check whether $<$ satisfies the Wheeler properties in optimal $O(|V|+|E|)$ time. 
\end{lemma}
%\begin{proof}[Sketch]
%This can be achieved easily by radix-sorting edges $(u,v,a)$ represented as triples $(a,i_u,i_v)$, where $i_u$ and $i_v$ are the Wheeler ranks of the source and the destination, and checking the Wheeler properties with one scan of the sorted triples.
%\end{proof}

%We now use Corollary \ref{lem:clusters} and Lemma \ref{lem: check range consistency} to prove the following (full proof in Appendix \ref{proof thm: DFA linear}): 

\begin{theorem}\label{thm:DFA linear}
	Let  $\mathcal A=(V,E,F,s,\Sigma)$ be a DFA. In $O(|V|+|E|)$ time we can:
	\begin{enumerate}
		\item Decide whether $\mathcal A$ is a Wheeler graph, and
		\item If $\mathcal A$ is a Wheeler graph, return a node ordering satisfying the Wheeler graph definition. 
	\end{enumerate}
\end{theorem}
\begin{proof}[Sketch]
	By Corollary \ref{lem:clusters}, if $\mathcal A$ is a Wheeler graph then we can use the strings labeling \emph{any} paths $s\rightsquigarrow u$ and $s\rightsquigarrow v$ to decide the order of any two nodes $u$ and $v$. 
	We build a spanning tree of $\mathcal A$ rooted in $s$ and prefix-sort it using~\cite[Thm 2]{ferragina2009compressing}. Finally, we verify correctness using Lemma \ref{lem: check range consistency}.\qed
\end{proof}

We note that the above strategy cannot be used to sort Wheeler NFAs, since the spanning tree could connect $s$ with several distinct nodes using the same labeled path: this would prevent us to find the order of those nodes using the spanning tree as support.

\section{Wheeler DFA Minimization}\label{sec:minimization}

We are now ready to use the algorithms of the previous sections to prove our main algorithmic results: (i) a minimization algorithm for WDFAs (Theorem \ref{thm: min DFA}) and (ii) a near-optimal algorithm generating the minimum acyclic WDFA equivalent to any input acyclic DFA (Theorem \ref{thm: ADFA -> WADFA}).

%\subsection{WDFA minimization}

Let $\equiv$ be an equivalence relation over the states $V$ of an automaton $\mathcal A = (V,E,F,s,\Sigma)$. The \emph{quotient automaton} is defined as $\mathcal A/_\equiv = (V/_\equiv, E/_\equiv, F/_\equiv,[s]_{\equiv},\Sigma)$, where $E/_\equiv = \{([u]_\equiv, [v]_\equiv, c)\ :\ (u,v,c)\in E\}$. The symbol $\approx$ denotes the Myhill-Nerode equivalence among states~\cite{nerode1958linear}: $u \approx v$, with $u,v\in V$, if and only if, for any string $\alpha$, we reach a final state by following the path labeled $\alpha$ from $u$ if and only if the same holds for $v$.
Note that this is the ``state'' version of the relation $\equiv_{\mathcal L}$ given in Section \ref{sec:WL} (which instead is defined among strings). 
The goal of any DFA-minimization algorithm is to find $\approx$, which is the, provably existing and unique, coarsest (i.e. largest classes) equivalence relation stable with respect to the initial partition in final/non-final states.
To abbreviate, we will simply say ``coarsest equivalence relation'' instead of ``coarsest equivalence relation stable with respect to an initial partition''.

In our case,  assuming that $\mathcal A$ is Wheeler, we want to find the (unique as proved below) coarsest equivalence relation $\equiv_w$ finer than $\approx$, such that $\mathcal A/_{\equiv_w}$ is Wheeler.  
Our Algorithm \ref{alg:minimize} achieves precisely this goal: we start with $\approx$ and then refine it preserving stability with respect to characters, while also ensuring that the resulting equivalence classes can be ordered consistently with the Wheeler constraints. 
Again, it can be proved that $\equiv_w$ is the ``state'' version of the relation $\equiv_{\mathcal L}^c$ given in Section \ref{sec:WL}.
For the purposes of the following results, we do not need to prove the connection between the two relations and we keep a distinct notation to avoid confusion.
We show (formal proof in Appendix \ref{proof thm: min DFA}):

\begin{theorem}\label{thm: min DFA}
	Let $\mathcal A$ be a WDFA.	The automaton $\mathcal{A}/_{\equiv_w}$ returned by Algorithm \ref{alg:minimize} is the minimum WDFA recognizing $\mathcal L(\mathcal A)$.
\end{theorem}

\begin{algorithm}
	\caption{WheelerMinimization($\mathcal A$)}\label{alg:minimize}
	\SetKwInOut{Input}{input}
	\SetKwInOut{Output}{output}
	
	\BlankLine
	\Input{Wheeler DFA $\mathcal A$}
	\Output{Minimum Wheeler DFA $\mathcal A'$ such that $\mathcal L(\mathcal A) = \mathcal L(\mathcal A')$}
	\BlankLine
	
	\begin{enumerate}
		\item Compute the Myhill-Nerode equivalence $\approx$ among states of $\mathcal A$.
		\item Prefix-sort $\mathcal A$'s states, obtaining the ordering $v_1 < \dots < v_n$.
		\item Compute a new relation $\equiv_{w}$ defined as follows. Insert in the same equivalence class all maximal runs $v_i < v_{i+1} < \dots < v_{i+t}$ such that:
		\begin{enumerate}
			\item $v_i \approx v_{i+1} \approx \dots \approx v_{i+t}$
			\item $\lambda(v_i) = \lambda(v_{i+1}) = \dots = \lambda(v_{i+t})$. 
		\end{enumerate}
		\item Return $\mathcal{A}/_{\equiv_w}$.
	\end{enumerate}
\end{algorithm}

%\begin{proof}[Sketch]
%	Consider the infinite automata tree $\mathcal T$ recognizing $\mathcal L(\mathcal A)$ obtained by ``unraveling'' $\mathcal A$. The minimum WDFA accepting $\mathcal L(\mathcal A)$ can be obtained by minimizing $\mathcal T$ without breaking the Wheeler properties.
%	By the same arguments of Lemma \ref{lem:clusters}, the co-lexicographic order of $\mathcal T$'s nodes is unique. Moreover, by the Wheeler properties we can only collapse states of $\mathcal T$ that are adjacent in co-lexicographic order. The next step is to prove that all of $\mathcal T$'s nodes that are collapsed in $\mathcal A$, must also be collapsed in the minimum WDFA. These observations imply that the minimum WDFA is obtained by collapsing maximal runs of $\approx$-equivalent states of $\mathcal A$ that are adjacent in co-lexicographic order and with the same incoming label. This is precisely what Algorithm \ref{alg:minimize} does. \qed
%\end{proof}

Note that uniqueness of the minimum WDFA follows from Corollary \ref{lem:clusters} (uniqueness of the Wheeler order) and Algorithm \ref{alg:minimize}.
Note also that, in the automaton output by Algorithm \ref{alg:minimize}, adjacent states in co-lexicographic order are distinct by the relation $\approx$ unless their incoming labels are different (in which case they might be equivalent). It follows that if a sorted Wheeler DFA does not have this property, then it is not minimum (otherwise Algorithm \ref{alg:minimize} would collapse some of its states). Conversely, If a Wheeler DFA has this property, then Algorithm \ref{alg:minimize} does not collapse any state, i.e. the automaton is already of minimum size. We therefore obtain the following characterization:

\begin{theorem}[Minimum WDFA]\label{thm:characterization minimum} Let $\mathcal A$ be a Wheeler DFA, let $v_1 < v_2 < \dots < v_t$ be its co-lexicographically ordered states, and let $\approx$ be the Myhill-Nerode equivalence among them. $\mathcal A$ is the minimum Wheeler DFA recognizing $\mathcal L(\mathcal A)$ if and only if the following holds: for every $1 \leq i < t$, if $v_i \approx v_{i+1}$ then $\lambda(v_i) \neq \lambda(v_{i+1})$.
\end{theorem}

Theorem \ref{thm: min DFA} implies the following corollaries.

\begin{corollary}\label{thm: min DFA1}
	Given a WDFA $\mathcal A$ of size $n$, in $O(n\log n)$ time we can build the minimum WDFA recognizing $\mathcal L(\mathcal A)$.
\end{corollary}
\begin{proof}
	We run Algorithm \ref{alg:minimize} computing $\approx$ with Hopcroft's algorithm~\cite{hopcroft1971n} ($O(n\log n)$ time), and prefix-sorting $\mathcal A$ with Theorem \ref{thm:DFA linear} ($O(n)$ time). Note that we can check $u\approx v$ in constant time by representing the equivalence relation as a vector $EQ[v] = [v]_{\approx}$, where we choose $V=\{1, \dots, |V|\}$ and where $[v]_{\approx}$ is any representative of the equivalence class of $v$ (e.g., the smallest one, which we can identify in linear time by radix-sorting equivalent states). Then, $u\approx v$ if and only if $EQ[u]=EQ[v]$. Using this structure, the runs of Algorithm \ref{alg:minimize} can easily be identified in linear time.  \qed
\end{proof}

\begin{corollary}\label{thm: min ADFA}
	Given an acyclic WDFA $\mathcal A$ of size $n$, in $O(n)$ time we can build the minimum acyclic WDFA recognizing $\mathcal L(\mathcal A)$.
\end{corollary}
\begin{proof}
	We run Algorithm \ref{alg:minimize} computing $\approx$ with Revuz's algorithm~\cite{revuz1992minimisation} ($O(n)$ time), prefix-sorting $\mathcal A$ with Theorem \ref{thm:DFA linear} ($O(n)$ time), and testing  $u\approx v$ in constant time as done in Corollary \ref{thm: min DFA1}. \qed
\end{proof}

Note that Corollary \ref{thm: min ADFA} implies that we can, in \emph{optimal linear} time, build the minimum WDFA $\mathcal A/_{\equiv_w}$ recognizing \emph{any} input finite language $\mathcal L$ represented as a set of strings: we build the tree DFA accepting $\mathcal L$ and apply Corollary \ref{thm: min ADFA}. The corollary can be applied since trees are always Wheeler~\cite{ferragina2009compressing,gagie2017wheeler}. 
In the next subsection we treat the (more interesting) case where $\mathcal L$ is represented by a DFA.
Note that this result could already be achieved by unraveling the DFA into a tree and minimizing it using Corollary \ref{thm: min ADFA}. However, the intermediate tree could be exponentially larger than the output. 

\subsection{Acyclic DFAs to Smallest Equivalent WDFAs}\label{sec:ADFA->WADFA}

We show how to build the smallest acyclic Wheeler DFA equivalent to any acyclic DFA in output-sensitive time.
Let $\mathcal A = (V,E,F,s,\Sigma)$ be an acyclic DFA. We first minimize $\mathcal A$ using Revuz's algorithm~\cite{revuz1992minimisation} and obtain the equivalent minimum acyclic DFA $\mathcal A_1 = \mathcal A/_\approx = (V_1,E_1,F_1,s_1,\Sigma)$.
%The first step is to make $\mathcal A'$ input-consistent: we replace each state $v\neq s$ with input label  set $\{c_1, \dots, c_k\}$ with $k$ states $v_{c_1}, \dots, v_{c_k}$ such that $v_c$ has only the incoming edges labeled $c$ of $v$.  We make each $v_{c_1}, \dots, v_{c_k}$ final if and only if $v$ is final. We moreover replace each outgoing edge $(v,u,c)$ of $v$ with $k$ edges $(v_{c_1},u_c,c), \dots, (v_{c_k},u_c,c)$. It is easy to see that the resulting input-consistent automaton $\mathcal A$ has at most $|V'|\sigma$ states and at most $|E'|\sigma$ edges and recognizes the same language as $\mathcal A'$. 
Let us denote $|V_1|=t$.
%The second step is to topologically-sort the states of $\mathcal A'$. Let $v_1, \dots, v_t$ be the resulting order (with $s'=v_1$). 
%The third step is to compute the equivalence class $\approx$ over $V$ corresponding to the minimum acyclic DFA recognizing $\mathcal L(\mathcal A)$. We can do this in linear time using Revuz's algorithm~\cite{revuz1992minimisation}.
The idea is to run a modified version of the online Algorithm \ref{alg:step} on $\mathcal A_1$. The difference is that now we will \emph{solve} (not just detect) violations to the Wheeler properties without changing the accepting language.
% instead of globally as discussed at the end of the previous section.
%We use three dynamic sequences $\mathtt{IN}$, $\mathtt{LEX}$, and $\mathtt{OUT}$ as done in Algorithm \ref{alg:step}: in this case, $\mathtt{LEX}$ will contain the co-lexicographically-sorted states of $\mathcal A/_{\equiv_w}$. The other two sequences will contain the incoming ($\mathtt{IN}[[u]_{\equiv_w}]$) and outgoing ($\mathtt{OUT}[[u]_{\equiv_w}]$) labels of each state $[u]_{\equiv_w} \in {V/\equiv_w}$. 
The next step is to topologically-sort $\mathcal A_1$'s states (e.g. using Kahn's algorithm~\cite{kahn1962topological}).
At this point, we modify $\mathcal A_1$ in $t$ steps by processing its states in topological order. This defines a sequence of automata $\mathcal A_1, \mathcal A_1, \dots, \mathcal A_t$. At each step, the states of $\mathcal A_i$ are partitioned in two sets:
\begin{itemize}
	\item those not yet processed: $N_i = \{ v_{i+1}, v_{i+2}, \dots, v_{t} \}$, and 
	\item the remaining states $V_i - N_i$, sorted by a total ordering $<$ in a sequence $\mathtt{LEX}_i$.
\end{itemize}
At the beginning, $N_1 = \{v_2, \dots, v_t\}$ and $\mathtt{LEX}_1 = s$. Note that $N_t = \emptyset$ (i.e. at the end we will have processed all states). At each step $i$, we maintain the following invariants:
\begin{enumerate}
	\item $\mathcal L(\mathcal A_i) = \mathcal L(\mathcal A_1)$.
	\item States in $\mathtt{LEX}_i$ are sorted by a  total order $<$ that does not violate the Wheeler properties among states in $\mathtt{LEX}_i$ itself: in Definition \ref{def_WG}, we require $u_1,u_2,v_1,v_2 \in \mathtt{LEX}_i$.
	\item for each $j=1, \dots, |\mathtt{LEX}_i|-1$, if $\mathtt{LEX}_i[j] \approx \mathtt{LEX}_i[j+1]$ then $\lambda(\mathtt{LEX}_i[j]) \neq \lambda(\mathtt{LEX}_i[j+1])$.
\end{enumerate}
Invariant {\bf 1} implies $\mathcal L(\mathcal A_t) = \mathcal L(\mathcal A)$. Since $N_t=\emptyset$ and $\mathtt{LEX}_t$ contains all $\mathcal A_t$'s states, invariant {\bf 2} implies that $\mathcal A_t$ is Wheeler (note that intermediate automata $\mathcal A_i$, with $1 < i < t$ might be non-Wheeler). Finally, invariant {\bf 3} and Theorem \ref{thm:characterization minimum} imply that $\mathcal A_t$ is the minimum WDFA accepting $\mathcal L(\mathcal A_t)$. As a result, $\mathcal A_t = \mathcal A/_{\equiv_w}$. We describe all the  details of our algorithm in Appendix \ref{app:DFA->WDFA} for space constraints; here we give an overview of the procedure. The idea is to process states in topological order as done in Theorem \ref{thm:n log n}. This time, however, we also solve inconsistencies of type 1 and 2 among nodes in $\mathtt{LEX}_{i} \cup \{v_{i+1}\}$ by splitting nodes in $\approx$-equivalent copies. 
Here, \emph{splitting} means creating two or more copies of a state $v$ in such a way that (i) each copy duplicates all $v$'s outgoing edges, (ii) $v$'s incoming edges are distributed (not duplicated) among the copies, and (iii) each copy is a final state if and only if $v$ is a final state. 
%Figure \ref{fig:make WDAG} shows how the inconsistency-resolution process is carried out when processing a node $v$ in order to insert it in $\mathtt{LEX}_i$.
%We may split either $v$ (in the figure, $v$ is splitted in $v_1,v_2,v_3$), to solve inconsistencies of type 1, or some $a$-successor (in the figure, $z_3$) of at least two nodes (in the figure, $w_3,w_4$) that are separated in co-lexicographic order by some $v$'s predecessor (in the figure, $w_3 < u_2 < u_3 < w_4$), to solve inconsistencies of type 2.
%We can make sure that the processed node $v$ has all edges labeled with the same character by first splitting it into one node per distinct incoming label, and then processing separately each resulting node: let $\{c_1, \dots, c_k\}$ be the set of $v$'s incoming labels.
%We replace $v$ with $k$ states $v_{c_1}, \dots, v_{c_k}$ such that $v_c$ has only the incoming edges labeled $c$ of $v$.  We make each $v_{c_1}, \dots, v_{c_k}$ final if and only if $v$ is final. We moreover replace each outgoing edge $(v,u,c)$ of $v$ with $k$ edges $(v_{c_1},u_c,c), \dots, (v_{c_k},u_c,c)$.
%First, note that 
Our splitting process creates $\approx$-equivalent nodes, therefore the accepted language never changes (invariant {\bf 1} stays true).
Moreover, since the states of $\mathcal A$ have already been collapsed by the equivalence $\approx$, after inserting nodes (or their copies) in $\mathtt{LEX}_{i}$ we never create runs of length greater than one of $\approx$-equivalent states with equal incoming labels (invariant {\bf 3} stays true). As a result, we incrementally build the minimum WDFA $\mathcal A/_{\equiv_w}$ recognizing $\mathcal L(\mathcal A)$. Since our algorithm never deletes edges, the running time is bounded by the output's size (which could nevertheless be much larger --- or smaller --- than $\mathcal A$). 
%In a sense, we ``unravel'' $\mathcal A_1$ \emph{locally}. 
In Appendix \ref{app:DFA->WDFA} we show:

\begin{theorem}\label{thm: ADFA -> WADFA}
	Given an acyclic DFA $\mathcal A$ of size $n$, we can build and prefix-sort the minimum acyclic WDFA, of size $m$, recognizing $\mathcal L(\mathcal A)$ in $O(n + m\log m)$ time.
\end{theorem}

Theorem \ref{thm: ADFA -> WADFA} solves the problem of indexing deterministic DAGs for linear-time pattern matching queries in nearly-optimal time with a solution of minimum size. 
Note that the hardness result of Equi et al.~\cite{equi2019complexity2} implies that, under the Orthogonal Vectors hypothesis, in the worst case the minimum WDFA has size $\Omega(n^{2-\epsilon})$ for any constant $\epsilon>0$. We can do better: in Appendix \ref{sec:worst case blow up} we show that, in the worst case, the minimum WDFA can be of size $\Omega(2^{n/4})$.

\appendix

\section{Indexing Wheeler Automata}\label{sec:indexing}

We show that any Wheeler NFA can be efficiently indexed in order to support fast membership queries in its accepting language or in its substring/suffix closure. 
Let $\mathcal A$ be any Wheeler NFA. We first remove all states that do not lead to a final state. This preserves the accepted language, the total ordering, and the Wheeler properties. We then use our algorithms to convert the automaton to a WDFA, prefix-sort it in polynomial time, and build a (generalized) FM-index on the graph as described in~\cite{gagie2017wheeler}.
We mark in a bitvector $B[1..|V|]$ supporting constant-time \emph{rank} queries~\cite{jacobson1988succinct} all accepting states of the Wheeler NFA in our array $\mathtt{LEX}$ containing the states in co-lexicographic order. To check membership of a word $w$, we search the word $\#w$ and get a range $\mathtt{LEX}[L,R]$ of all states reachable from the root by a path labeled $w$. At this point, $w$ is accepted if and only if $B[L,R]$ contains at least one bit set (constant time using \emph{rank} on $B$). Note that this procedure works in $O(w\log\sigma)$ time also if the original automaton is nondeterministic (this, in general, is not possible for general NFAs). If we search for $w$ instead of $\#w$, then we get the range of states reachable by a path labeled $uw$, for any $u\in\Sigma^*$. This range is non-empty if and only if $w$ belongs to the substring closure of $\mathcal L(\mathcal A)$. Finally, if we search a word $w$ and get a range $\mathtt{LEX}[L,R]$, then $w$ is in the suffix closure of $\mathcal L(\mathcal A)$ if and only if $B[L,R]$ contains at least one bit set. 

\section{Conclusions and Future Extensions}

In this paper, we have initiated the study of Wheeler languages, that is, regular languages that can be indexed via prefix-sorting techniques. On our way, we provided new results of independent interest: (i) we provided a new class of NFAs for which the minimization problem can be approximated up to multiplicative factor 2 in polynomial time and that admit fast membership and pattern matching algorithms, and (ii) we solved the problem of indexing finite languages with prefix-sortable DFAs of minimum size. Our work leaves several intriguing lines of research (some of which will be explored in future extensions of this paper). First of all, is the problem of recognizing Wheeler languages (encoded, e.g. as regular expressions) decidable? We believe that the answer to this question is positive: the \emph{Wheelerness} of a regular language seems to translate into particular constraints (that can be verified in bounded time) on the topology of its minimum accepting DFA. Once a regular language has been classified as Wheeler, can we build the minimum accepting Wheeler DFA? Also in this case, we believe that the task can be solved by iterating conflict-resolution (Section \ref{sec:ADFA->WADFA}) from the minimum DFA until the process converges to the minimum Wheeler DFA.

\section{Proofs of Section \ref{sec:WL}}

\begin{proof} (of Lemma \ref{3n})

Let us order the elements of $L$ by the relation $<$, and let us denote by $L[i]$ the $i$-th element in the ordering. The notation $L[i,j]$, with $j\geq i$, denotes the interval $\{L[k]:i\leq k \leq j\}$. In particular, $L[1,n] = L$.

We say that an interval $I\in \mathcal C$ is \emph{maximal} if $I$ is not the prefix nor the suffix of any other interval in $\mathcal C$. 
We say that an interval $I\in \mathcal C$ is \emph{prefix} (resp. \emph{suffix}) if $I$ is the \emph{proper} prefix (resp. suffix) of a maximal interval of $\mathcal C$. 
Note that, by this definition, intervals of $\mathcal C$ are either maximal or prefix/suffix. Note also that there could be elements of $\mathcal C$ being both prefix and suffix intervals.

(1) We first prove that (1.i) $\mathcal C$ contains at most $n$ prefix intervals, then (1.ii) slightly improve this bound to $n-1$, and finally (1.iii) show that the sum between the number of maximal and suffix intervals is at most $n$.
To prove (1.i), we show that every $L[j]$, $1\leq j \leq n$, can be the largest element of at most one prefix interval. In turn, this is shown by considering the prefixes of any two pairwise distinct maximal intervals of $\mathcal C$. 
Consider two distinct maximal intervals $I=L[i,j]$ and $J=L[i',j']$. If $I$ and $J$ do not overlap (i.e. $j < i'$ or $j' < i$), then the property is trivially true: if $I'$ and $J'$ are prefixes of $I$ and $J$, respectively, then $\max(I') \neq \max(J')$. Consider now the case where $I$ and $J$ overlap. Without loss of generality, we can assume $i < i' \leq j < j'$ (the strict inequalities follow from the fact that, by maximality, it cannot be $i=i'$ or $j=j'$). Assume, for contradiction, that $\mathcal C$ contains two intervals $L[i,j'']$ and $L[i',j'']$ such that $i' \leq j'' < j$, i.e. $L[i,j'']$ and $L[i',j'']$ are (proper) prefixes of $I$ and $J$, respectively, that share their largest element $j''$. Then, we have $i < i' \leq j'' < j$: interval $L[i',j'']$ is strictly contained inside $L[i,j]\in\mathcal C$ (i.e. $L[i',j''] \subset L[i,j]$) and it is not a prefix nor a suffix of it. This is forbidden by the definition of prefix/suffix family. From this contradiction we deduce that any two distinct prefix intervals $I', J' \in \mathcal C$
satisfy $\max(I') \neq \max(J')$, which implies that $\mathcal C$ contains at most $n$ prefix intervals. 

To improve the above bound to $n-1$ and prove (1.ii), consider the rightmost maximal interval $I=L[i,j]$, i.e. the one having largest $j$. We show that $j$ cannot be the maximum element of any prefix interval. Assume, for contradiction, that such a prefix interval $K=L[i',j]$ exists. Then, the corresponding maximal interval $J = L[i',j']$ of which $K$ is a proper prefix satisfies $j'>j$. This contradicts the fact that $I$ is the rightmost maximal interval. 

The next step is to prove (1.iii), i.e. that the sum between the number of maximal and suffix intervals is at most $n$. We proceed by induction on the number $M$ of maximal intervals. If $M=1$, then the unique maximal interval $I=L[i,j]$ contains at most $j-i$ suffix intervals. In total, $\mathcal C$ contains at most $1 + (j-i) \leq n$ maximal and suffix intervals. For $M>1$, consider the maximal interval $I = L[i,j]$ with minimum $j$ (call it the ``leftmost''). Now, consider the immediate maximal ``successor'' $J = L[i',j']$ of $I$, i.e. the maximal interval with the smallest endpoint $j'\geq j$. 
Clearly, such $j'$ satisfies $j'>j$, otherwise $J$ would be a suffix of $I$ (contradicting maximality of $J$).
Note that it must also be the case that $i'>i$: if $i=i'$, then $I$ would be a prefix of $J$ (contradicting maximality of $I$); on the other hand, if $i'<i$ then $I$ would be strictly contained in $J$, contradicting the definition of prefix/suffix family. We are left with two cases: 

(a) $i \leq j < i' \leq j'$. In this case, $I$ and $J$ are disjoint. As seen above, $I$ contributes to at most one maximal interval ($I$ itself) and $j-i$ suffix intervals. In total, $I$ contributes to at most $j-i+1$ maximal and suffix intervals. We are left to count the number of maximal and suffix intervals in the remaining portion of the linear order $L[i',...,n]$. Note that there are no other intervals to be considered: if $L[i'', j'']$ is a maximal interval in  $\mathcal C$, different from $I,J$, 
then $j''>j'$ and hence $i''>i'$ or $L[i'',j'']$ would contain $J$.   The portion $L[i',...,n]$ contains $M-1$ maximal intervals, so we can apply the inductive hypothesis and obtain that this segment contains at most $n-i'+1$ maximal and suffix intervals. In total, we have that $L[1,...,n]$ contains at most $(j-i+1) + (n-i'+1)$ maximal and suffix intervals. Since $i'>j$ and $i\geq 1$, this quantity is at most $n$. 

(b) $i <  i' \leq j < j'$. Denote by $k = i'-i$ the number of $L$'s elements belonging to $I \setminus J$. Then, $\mathcal C$ can contain at most $k$ proper suffixes of $I$: $L[i+1,j]$, $L[i+2,j]$, ..., $L[i', j]$. All other suffixes of $I$ are strictly contained inside $J$, and cannot belong to $\mathcal C$ due to the prefix/suffix property. Actually, one of those suffixes, $L[i',j]$, is a prefix of $J$ so it has already been counted above in points (1.i) and (1.ii). We are left with $k-1$ suffixes to take into account, plus the maximal interval $I$ itself: in total, $k = i'-i$ maximal and suffix intervals. As noted above, all remaining  maximal and suffix intervals of $\mathcal C$ to take into account are those contained in $L[i',n]$.  Since $L[i',n]$ contains $M-1$ maximal intervals, we can apply the inductive hypothesis and deduce that it contains at most $n-i'+1$ maximal and suffix intervals. In total, $L[1,n]$ contains therefore at most $(i'-i) + (n-i'+1) \leq n$ maximal and suffix intervals. This concludes the proof of the upper bound $|\mathcal C| \leq 2n-1$.

(2) Consider the prefix/suffix family containing just one maximal interval and all its proper prefixes and suffixes:
$\mathcal C = \{  L[1,n], L[1,1], \dots, L[1,n-1], L[2,n], \dots, L[n,n]\}$. This family satisfies $|\mathcal C| = 2n-1$. \qed
\end{proof}

 \begin{proof} (of Lemma \ref{convex+order})
 We just prove transitivity when  $I<^{i} J$ and  $J<^{i} K$  are  witnessed by $x_0\in I$ satisfying   $ (\forall y\in J)  ( x_0< y )$,  and $  z_0 \in K$ satisfying    $ (\forall y \in J) ( y< z_0)$, respectively (the other cases are similar).  We claim that $z_0>x$, for all $x\in I$. Suppose, for contradiction, that   there exists $x_1\in I$ with $z_0\leq x_1$;  then, from $x_0<y<z_0\leq x_1$ for all $y\in J$   and the fact that $I$ is an interval, it follows that $z_0\in I$,  $J\subseteq I$,  so that, by   prefix/suffix property of $\mathcal C$,    $J$ is either a prefix or a suffix of $I$.   Since $ x_0< y$ for all $y\in J$, we see that $J$ must be a suffix of $I$  and this, knowing that $z_0\in I$,  implies   $z_0\in  J$. A contradiction.  
 \end{proof}

In the following proofs, we always refer to a WNFA $\mathcal A=(V,E,F,s, \Sigma,<)$.
 \begin{proof}  (of Lemma \ref{prec_versus_minus})
\begin{enumerate}
\item[(1)]  Suppose   $\alpha\in I_u, \beta\in I_v$ and  $\{\alpha,\beta\}\not\subseteq I_v\cap I_u$. From this we have that  $\alpha\in I_u\setminus I_v$  or   $\beta\in I_v\setminus I_u$,  hence  $u\neq v$ and $\alpha\neq \beta$ follows.

If $u=s$ or $v=s$, either $ \alpha $ or $ \beta $ is the empty string $ \epsilon $ and   the result follows easily. Hence, we suppose $u\neq s \neq v$ and (hence)  $\alpha\neq \epsilon \neq \beta$. 
 
To see the left-to-right implication, assume $\alpha \prec \beta$: we prove that $u<v$ 
 by induction on the maximum betwewn $|\alpha|$ and $|\beta|$. If $|\alpha| =|\beta|=1$, then the property follows from the Wheeler-(i). If $\max(|\alpha|, |\beta|)>1$ and $ \alpha $ and $ \beta $  end with different letters, then again the property follows from  Wheeler-(i). Hence, we are just left with the case in which $\alpha=\alpha' e$ and $\beta=\beta' e$, with $e\in \Sigma$.
If $\alpha\prec \beta$,  then  $\alpha'\prec\beta'$. Consider  states
$u',v'$ such that $\alpha' \in I_{u'}, \beta'\in I_{v'}$, and $(u',u,e), (v',v,e)\in E$. 
Then $\alpha'\in I_{u'}\setminus I_{v'}$ or  $\beta'\in I_{v'}\setminus  I_{u'}$ because otherwise we would have  $\alpha'\in I_{v'}$   and $\beta'\in I_{u'}$   which imply respectively $\alpha\in I_v$ and $\beta\in I_u$.  By induction we have $u'<v'$ and therefore, by Wheeler-(ii), $u\leq v$. From $u\neq v$ it follows  $u<v$.   

 Conversely, for the right-to-left implication, suppose   $u<v$.  Since  $\alpha \neq \beta$,  if it were $\beta\prec \alpha$ then, by the above, we would have $v<u$:  a contradiction. Hence, $\alpha\prec \beta $ holds. 
\item[(2)]  Recall that, by definition,  $ \alpha \in I_u  \text{ if and only if } u \in I_\alpha$ and $ \beta \in I_v  \text{ if and only if } v \in I_\beta$. Hence, the hypothesis that $u\in I_\alpha, v\in I_\beta$ and $\{u,v\} \not\subseteq I_\beta \cap I_\alpha$, is equivalent to say that $ \alpha\in I_u, \beta\in I_v$ and $\{\alpha,   \beta\} \not\subseteq I_v \cap I_u$.
  Therefore, (2) follows from (1).
\end{enumerate} 
 \end{proof}

 \begin{proof} (of Lemma \ref{convex_sets})
 \begin{enumerate}
\item  Suppose $\alpha\prec \beta\prec\gamma$ with $\alpha, \gamma\in I_u$ and $\beta \in  \text{\em Pref}(\mathcal L(\mathcal A))$;  we want to prove that $\beta \in I_u$. From $\beta\in  \text{\em Pref}(\mathcal L(\mathcal A))$ it follows that  there exists  a state $v$ such that
 $\beta \in I_v$.  Suppose, for contradiction, that $\beta\not \in I_u$. Then $\beta \in I_v\setminus I_u$ and from $\alpha\prec \beta$ and  Lemma \ref{prec_versus_minus}, it follows $u<v$. Similarly, applying again Lemma \ref{prec_versus_minus}, from $\beta\prec \gamma$ we have  $v<u$, which is a contradiction. 
 \item   Suppose, for contradiction, that  $I_u, I_v\in  I_{V}$ are such that $I_u\subsetneq I_v$ and  $I_u$ is neither a prefix nor a suffix of $I_v$.  In these hypotheses  there
must exist $\alpha, \alpha' \in I_v\setminus I_u$ and $\beta \in I_u$ such that  $\alpha \prec \beta \prec \alpha'$. Lemma \ref{prec_versus_minus}  implies
$v<u<v$, which is a contradiction. 
 \end{enumerate}
 Points $(3), (4)$ follow  similarly from Lemma  \ref{prec_versus_minus}. 
 \end{proof}
 
 \begin{proof} (of Lemma \ref{Wdeterminization})
 
The verification that $ \mathcal L(\mathcal A^{d}) = \mathcal L(\mathcal A)$ follows the same lines of the proof in the classical regular case.   
We  prove that $<^{d}$ is a Wheeler order on the states of the automaton $\mathcal A^{d}$.  By Lemma  \ref{convex_sets}, the set  $V^{d}=I_{\text{\em Pref}(\mathcal L(\mathcal A))}$ of states of  $\mathcal A^{d}$ is a prefix/suffix family of intervals, so that,  by   Lemma \ref{convex+order},    $<^{d}$ is a linear order on $V^{d}$.
Next, we check the Wheeler properties.
The only vertex with  in-degree $0$  is $I_\epsilon$, and it clearly precedes those with positive in-degree.
For any two edges  $(I_\alpha,I_{\alpha a_1}, a_1)$,  $(I_\beta,I_{\beta a_2}, a_2)$ we have:
	\begin{itemize}
		\item[(i)]  if $a_1 \prec a_2$  then  $\alpha a_1\prec \beta a_{2}$,  and  from Lemma \ref{compare} it follows   $I_{\alpha a_1}\leq^d I_{\beta a_2}$.  Moreover, by the input consistency of  $\mathcal A$, states in $I_{\alpha a_1}$ are $a_1$-states, while states in $I_{\beta a_2}$ are $a_2$-states; hence   $I_{\alpha a_1}\neq I_{\beta a_2}$, so that $I_{\alpha a_1}<^d I_{\beta a_2}$ follows.
		\item[(ii)]   If $a=a_{1}=a_{2}$ and  $I_\alpha < I_\beta$,   from Lemma \ref{compare} it follows $\alpha\prec \beta$, so that $\alpha a \prec \beta a$ and, using again   Lemma \ref{compare}, we obtain   $I=I_{\alpha a}\leq^i I=I_{\beta a}$.
\end{itemize}
Finally, we prove that $|V^d|\leq 2n-1-|\Sigma|$. By the Wheeler properties, we know that the only interval in $I_{\text{\em Pref}(\mathcal L(\mathcal A))}$  containing the initial state $s$ of the automaton $\mathcal A$ is $\{s\}$ and that the remaining intervals can be partitioned into $|\Sigma|$-classes, by looking at the letter labelling incoming edges. 
If $\Sigma=\{a_1, \ldots, a_k\}$, and, for every $i=1,\ldots,k$, we denote by $m_i$    the number of states of   the automaton $\mathcal A$  whose incoming edges are labeled $a_i$, we have  
$\sum_{i=1}^k m_i=n-1$.  Using Lemma \ref{3n} we see that the intervals in $V^d$ composed by  $a_i$ states are at most $2m_i-1$, so that the total number of intervals in $V^d$ is at most 
\[1+ \sum_{i=1}^k (2m_i-1)=1+ 2( \sum_{i=1}m_i )-k= 1+2(n-1)-k=2n-1-k= 2n-1-|\Sigma|.\]
  	
 \end{proof}

 \begin{proof}[of Lemma \ref{ComputeWdeterminization}]
 
 We apply the standard powerset construction algorithm starting from the original WNFA $\mathcal A$.
 By Lemma \ref{Wdeterminization}, 
 the powerset algorithm does not generate more than $2n-1-|\Sigma|$ distinct sets of states. Remember that such algorithm starts from the set $\{s\}$ containing the NFA's source and simulates a visit of the final DFA, whose states are represented as sets of states of the original NFA. At each step, the successor with label $a\in\Sigma$ of a set $K$ is computed by calculating all the $a$-successors of states in $K$, and taking their union. In the worst case, $|K| = O(n)$ and each state in $K$ has $O(n)$ $a$-successors. After having obtained the $a$-successor $K'$ of $K$, we need to check if $K'$ had already been visited. Since $K'$'s cardinality is at most $n$, this operation takes $O(n)$ time using a standard dictionary (e.g. a hash table). Overall, we spend $O(n^2)$ time to simulate an edge traversal of the final DFA. By Lemma \ref{Wdeterminization}, we visit at most $O(n)$ distinct sets of states. Overall, the powerset algorithm's complexity is $O(n^3)$. \qed
 \end{proof}

\begin{proof} (of Theorem \ref{Myhill-Nerode})

\begin{itemize}
  \item[] (1) $\Rightarrow$ (2)  If  $\mathcal A$ is   a  Wheeler NFA  such that $\mathcal L=  \mathcal L(\mathcal A)$,  consider the following    equivalence relation $\sim_{\mathcal A}$ over $ \text{\em Pref}(\mathcal L)$:
  \[\alpha \sim_{\mathcal A} \beta ~~\Leftrightarrow I_\alpha=I_\beta.\]
  Using the fact that the  $I_\alpha$ are intervals (see Lemma \ref{convex_sets}), and other properties of Wheeler automata,  one can easily prove that the equivalence $\sim_{\mathcal A}$ is a refinement of  $\equiv_{\mathcal L}^{c}$, so that each $\equiv_{\mathcal L}^{c}$-class is a union of $\sim_{\mathcal A}$-classes. Moreover, the equivalence $\sim_{\mathcal A}$ has finite index, bounded by the number of intervals $I_\alpha$, hence  $\equiv_{\mathcal L}^{c}$ has finite index as well. 
  \item[] (2) $\Rightarrow$ (3)    We prove that the relation  $\equiv_{\mathcal L}^{c}$ is a convex, input consistent, right invariant equivalence, and that  $\mathcal L$ is a union of $\equiv_{\mathcal L}^{c}$-classes; this last property is true simply  because   $\mathcal L$ is a union of $\equiv_{\mathcal L}$-classes and 
  $\equiv_{\mathcal L}^{c}$ is a refinement of $ \equiv_{\mathcal L}$. The fact that $\equiv_{\mathcal L}^{c}$ is   convex and   input consistent     follows directly from its definition.
  We prove that $\equiv_{\mathcal L}^{c}$ is right invariant.  Suppose  $\alpha , \alpha' , \gamma\in   \text{\em Pref}(\mathcal L)$ and $\alpha\equiv_{\mathcal L}^{c} \alpha'$.   Note that
 if $\alpha \gamma \in   \text{\em Pref}(\mathcal L)$ then  there exists $\nu \in \Sigma^*$ such that  $\alpha \gamma \nu \in  \mathcal L$,  so that  $\alpha' \gamma \in  \text{\em Pref}(\mathcal L)$ follows from $\alpha\equiv_{\mathcal L} \alpha'$. Hence, we are left to prove that $\alpha\gamma \equiv_{\mathcal L} ^{c} \alpha'\gamma$. 
 We easily prove the following:
 \begin{itemize}
\item[-]  $\alpha \gamma \equiv_{\mathcal L}  \alpha'\gamma$ (it follows from $\alpha\equiv_{\mathcal L}\alpha'$). 
\item[-] If $\alpha \gamma \prec \beta' \prec \alpha' \gamma$,  for $\beta'\in   \text{\em Pref}(\mathcal L)$, then   $\beta'  \equiv_{\mathcal L} \alpha \gamma$: from  $\alpha \gamma \prec \beta' \prec \alpha' \gamma$ it follows that    $\beta'=\beta \gamma$,  for some $\beta\in  \text{\em Pref}(\mathcal L)$,    and $\alpha \prec \beta \prec \alpha'$. Since $\alpha, \alpha'$ belong to the same $\equiv_{\mathcal L}^{c} $ class, then     $\beta\equiv_{\mathcal L} \alpha$, and $\beta \gamma \equiv_{\mathcal L} \alpha \gamma$ follows from the right invariance of $\equiv_{\mathcal L} $.
\end{itemize}
Since $\alpha \gamma, \beta \gamma $ end with the same letter, the previous points imply  that  $\alpha\gamma \equiv_{\mathcal L} ^{c} \alpha'\gamma$ and  $\equiv_{\mathcal L}^{c} $ is right invariant.  
  
\item[] (3) $\Rightarrow$ (4)   Suppose   $\mathcal L$ is a union of classes of  a convex, input consistent,  right invariant  equivalence  relation $ \sim $  of finite index. We build a WDFA    ${\mathcal A}_\sim=(V_{\sim}, E_{\sim}, F_{\sim}, s_{\sim}, \Sigma,  <_{\sim} )$ such that $\mathcal L=\mathcal L(\mathcal A)$ as follows: 
\begin{itemize}
\item[-] $ V_{\sim}=\{[\alpha]_{\sim}~|~\alpha \in  \text{\em Pref}(\mathcal L)\}$;
\item[-] $ s_{\sim}=[\epsilon]_{\sim} $ (note that, by input consistency,  $[\epsilon]_{\sim}=\{\epsilon\}$);  
\item[-] $ (I,J,e) \in E_{\sim}$ if and only if $ Ie \cap \text{\em Pref}(\mathcal L) \neq \emptyset $ and $ Ie \subseteq J $, where $Ie=\{\alpha e ~|~ \alpha \in I\}$  (note that $ J $, if existing, is unique by right invariancy);
\item[-] $ F_{\sim}= \{I ~|~I \subseteq \mathcal L\}$;
\item[-] $<_{\sim} = \prec^{i}$ (being pairwise disjoint and convex,  the classes   in $V_{\sim}$ form a prefix/suffix family of intervals  of  $( \text{\em Pref}(\mathcal L), \prec)$).    
\end{itemize}
%For all $I\in V_{\sim}$ and $\alpha \in \text{\em Pref}(\mathcal L)$, we define $ \hat\delta_{\sim}(I,\alpha) $ inductively on $ |\alpha| $ as usual, and observe that $ \hat\delta_{\sim}(I,\alpha) $ is always a singleton set (i.e. $ \mathcal A_{\sim} $ is deterministic).  

\medskip

Note that   all words in $\text{\em Pref}(\mathcal L)$  label a computation in ${\mathcal A}_\sim$. 
We claim that, for all $\sim$-class $I$ and $\alpha\in\text{\em Pref}(\mathcal L)$:
\[
\alpha\in I  ~~\Leftrightarrow ~~ s_\sim \rightsquigarrow I \text{ in } \mathcal{A}_{\sim} \text{ reading } \alpha.\]

We  prove the implication from right to left  by induction on the length of  $\alpha\in  \text{\em Pref}(\mathcal L)$.
 
 If  $\alpha=\epsilon$ then the claim follows from the definition of $s_\sim$. 

If $\alpha=\alpha' e\in  \text{\em Pref}(\mathcal L)$ with  $e\in \Sigma$,  then     $\alpha'\in  \text{\em Pref}(\mathcal L)$.  Then, if $K\in V_{\sim}$ is such that  $s_\sim \rightsquigarrow K$ reading $\alpha'$ in $\mathcal{A}_{\sim}$,  by induction we know that $\alpha'\in K$. 
Since  $\alpha =\alpha' e\in  Ke $,  we have   $ Ke\cap \text{\em Pref}(\mathcal L) \neq \emptyset$;  by right invariance of $\sim$ there exists a unique $J$  such that  $Ke\subseteq  J$.  From
$\alpha=\alpha'e\in Ke\subseteq J$ it follows $\alpha\in J$, and also  $J=I$,    because ${\mathcal A}_\sim$ is a deterministic automaton and $s_\sim \rightsquigarrow I$, $s_\sim \rightsquigarrow J$,  both by reading $\alpha$. 

In order to prove the  implication from left to right of the claim,  suppose $\alpha \in I$, and  $J\in  V_\sim$  is such that
 $s_\sim \rightsquigarrow J \text{ in } \mathcal{A}_{\sim}$ reading $\alpha$. Then, by the first part of the proof of the claim we obtain $\alpha \in J$; since $J$ and $I$ are equivalence classes and $\alpha \in I \cap J$, it follows that $I=J$ and $s_\sim \rightsquigarrow I \text{ in } \mathcal{A}_{\sim}$ reading $\alpha$. 
 
 \medskip

From the above claim and the definition of $ F_{\sim} $, it easily follows that $\mathcal L$ is the language recognised by  ${\mathcal A}_{\sim}$.

\medskip 

We conclude by checking that  ${\mathcal A}_{\sim}$ is Wheeler, proving the two Wheeler properties (i) and (ii) with respect to the linear order $(V_\sim, <_{\sim})$. 

\medskip

 To see Wheeler-(i) assume  $e\prec e' $ with  $e,e'\in\Sigma$. Consider  $I,J \in V_{\sim}$ such that $(I,H,e)\in E_{\sim}$ and $(J,K,e')\in E_{\sim}$. We want to prove that  
 $H<_{\sim} K$ (i.e. $H\prec^{i} K$).   By definition of $E_{\sim}$, in our hypotheses there are  $\alpha\in I$, $ \alpha'\in J$ with  $\alpha e \in H$ and $ \alpha' e' \in K$. From  $e\prec e'$ it follows
 $\alpha e \prec \alpha' e'$  and hence  $H \preceq^{i}  K$. To conclude observe that  $H \prec^{i} K$ since all words in $H$ end with $e$, while all words in $K$ end with $e'$.

To see Wheeler-(ii) assume  $I<_{\sim} J$ (i.e. $I \prec^{i} J$), $e\in \Sigma$,   $(I,H,e)\in E_{\sim}$, and $(J,K,e)\in E_{\sim}$. In these hypotheses there are  $\alpha\in I$,  $\alpha'\in J $, with  $\alpha e \in H$ and $ \alpha' e \in K$.  From $I \prec^{i} J$ and the fact that different classes   are disjoint it follows $\alpha \prec \alpha'$;  therefore, $\alpha e \prec \alpha' e$ and  hence  $ H \preceq^{i} K $. 
  
 This ends the proof of the implication $(3) \Rightarrow (4)$. 

\medskip

\item[](4) $\Rightarrow$ (1) Trivial. 
\end{itemize}  
\end{proof}

\section{Proofs of Section \ref{sec:sorting}}

\begin{proof} (of Theorem \ref{thm:2NFA in P}) \label{proof thm 2NFA in P}

	We can assume, without loss of generality, that $\mathcal A$ is input-consistent, since checking this property takes linear time. If $\mathcal A$ is not input-consistent, then it is not Wheeler. 
	We show a reduction of problem {\bf 1} to 2-SAT, which can be solved in linear time using Aspvall, Plass, and Tarjan's (APT) algorithm based on strongly connected components computation. The reduction introduces $O(|V|^2)$ variables and $O(|E|^2)$ clauses, hence the final running time will be $O(|E|^2)$. Moreover, since a satisfying assignment to our boolean variables will be sufficient to define a total order of the nodes, APT will essentially solve also problem {\bf 2}. 
	
	For every pair $u\neq v$ of nodes we introduce a variable $x_{u<v}$ which, if true, indicates that $u$ must precede $v$ in the ordering. We now describe a 2-SAT CNF formula whose clauses are divided in two types: clauses of the former type ensure that the Wheeler graph property is satisfied, while clauses of the second type ensure that the order of nodes induced by the variables is total. 
	
	The following formulas ensure that the Wheeler properties are satisfied:
	
	\begin{itemize}
		\item[(a)] For each $u,v$, if $\lambda(u)\prec \lambda(v)$ then we add the unary clause $x_{u<v}$.
		\item[(b)]  For each $u\neq v$, if $\lambda(u)=\lambda(v)=a$, then for every pair $u'\neq v'$ such that $(u',u,a)\in E$ and $(v',v,a)\in E$ we add the clause $x_{u'<v'}\rightarrow x_{u<v}$.
	\end{itemize}
	
	There are at most $|V|^2\leq |E|^2$ clauses of type (a) and at most $|E|^2$ clauses of type (b).
	
	The following formulas guarantee that the order is total. Note that we omit transitivity which, on a general graph, would require a 3-literals clause $(x_{u<v} \wedge x_{v<w}) \rightarrow x_{u<w}$ for each triple $u,v,w$. We will show that, if the graph is an input-consistent $2$-NFA, then transitivity is satisfied ``for free''.
	
	\begin{itemize}
		\item[(1)] \emph{Antisymmetry}. For every pair $u\neq v$, add the clause $x_{u<v} \rightarrow \neg x_{v<u}$.
		\item[(2)] \emph{Completeness.} For every pair $u\neq v$, add the clause $x_{u<v}\ \vee\ x_{v<u}$. 
	\end{itemize}
	
	There are at most $O(|V|^2) = O(|E|^2)$ clauses of types (1) and (2).
	
	We now show that on input-consistent $2$-NFAs transitivity propagates from the source to all nodes. Consider a variable assignment that satisfies clauses (a),(b),(1), and (2) (if $\mathcal A$ is a Wheeler $2$-NFA, then such an assignment exists by definition). Assume, moreover, that $x_{u<v}$ and $x_{v<w}$ are set true by the assignment, for three pairwise distinct nodes $u,v,w$. We want to show that also $x_{u<w}$ must be true. 
	
	Consider a directed shortest-path tree $\mathcal T$ with root $s$ of $\mathcal A$. 
	Since we assume that each state is reachable from $s$, $\mathcal T$ must exist and must contain all nodes of $\mathcal A$. 
	Let $d_v$ be the length of a shortest directed path connecting $s$ to $v$. By definition of $\mathcal T$, the path connecting $s$ to $v$ in $\mathcal T$ has length $d_v$, with $d_s=0$.
	We proceed by induction on $k=\max\{ d_u,d_v,d_w \}$. The case $k=0$ is trivial, since there are no triples of pairwise distinct nodes in $\{u\ :\ d_u\leq 0\}$ (this set contains just $s$).
	Take now a general $k>0$. We consider two main cases: 
	
	(i) $|\{\lambda(u),\lambda(v),\lambda(w)\}|>1$. Then, since $x_{u<v}$ and $x_{v<w}$, for some $a<b<c\in\Sigma$ either: (i.1) $\lambda(u)=a,\ \lambda(v)=b,\  \lambda(w)=c$, or (i.2) $\lambda(u)=a,\ \lambda(v)=a,\ \lambda(w)=b$, or (i.3) $\lambda(u)=a,\ \lambda(v)=b,\ \lambda(w)=b$. Any other choice would force one of the variables $x_{v<u}, x_{w<v}$ to be true (by an (a)-clause), forcing a contradiction by a (1)-clause. In all cases (i.1)-(i.3) we have that $\lambda(u)<\lambda(w)$, therefore $x_{u<w}$ must be true by  (a).
	
	(ii) $\lambda(u) = \lambda(v) = \lambda(w) = a$ for some $a\in\Sigma$ (note that $a\neq \#$ since the NFA has only one source and $u,v,w$ are distinct by assumption). 
	Let $u',v',w'$ be the parents of $u,v,w$, respectively, in $\mathcal T$.
	Note that $u',v',w'$ cannot be the same vertex, since $u,v,w$ are distinct and every node has at most two outgoing edges with the same label. 
	We therefore consider two sub-cases. 
	
	(ii.1) $|\{u',v',w'\}| = 2$. 
	We first show that $u'=w'\neq v'$ generates a contradiction. Since $x_{u<v}$ and $x_{v<w}$ are true and $u'\neq v'$ and $v'\neq w'$ hold, $x_{u'<v'}$ and $x_{v'<w'}$ must be true: otherwise, by (b), would imply that $x_{v<u}$ and $x_{w<v}$ are true, which generates a contradiction. 
	Now, $u'=w'$ means that $x_{v'<w'}$ and  $x_{v'<u'}$ have the same truth value; since $x_{u'<v'}$ and $x_{v'<u'}$ cannot be both true, we have a contradiction.
	We are therefore left with the case $u'=v'\neq w'$ ($u'\neq v' = w'$ is symmetric). Remember that we assumed $x_{u<v}$ and $x_{v<w}$ are true. 
	Hence, $x_{v'<w'}$ must be true: otherwise, by (b), the truth of $x_{w'<v'}$ would imply that $x_{w<v}$ is true, which generates a contradiction. Since $x_{v'<w'} = x_{u'<w'}$ is true, by  (b) we conclude that also $x_{u<w}$ must be true.
	
	(ii.2) $u',v',w'$ are pairwise distinct.  We show that $x_{u'<v'}$ and $x_{v'<w'}$ must be true. 
	Suppose, for contradiction, that $x_{u'<v'}$ is false (the proof is symmetric for $x_{v'<w'}$). Then, by  (2), $x_{v'<u'}$ is true. But then, by (b) it must be the case that $x_{v<u}$ is true. Since we are assuming that $x_{u<v}$ is true, this introduces a contradiction by (1). Therefore, we conclude that $x_{u'<v'}$ and $x_{v'<w'}$ are true for the (pairwise distinct) parents $u',v',w'$ of $u,v,w$ in $\mathcal T$. Now, by definition of the shortest-path tree $\mathcal T$ it must be the case that $d_{u'} = d_{u}-1$, $d_{v'} = d_{v}-1$, and $d_{w'} = d_{w}-1$ as $u',v',w'$ are the parents of $u,v,w$ in $\mathcal T$. 
	%In particular, note that $(u',u), (v',v), (w',w)\in E$.
	As a consequence, $\max\{d_{u'},d_{v'},d_{w'}\} = k-1$. We can therefore apply the inductive hypothesis and conclude that $x_{u'<w'}$ is true. But then, by  (b) we conclude that $x_{u<w}$ must also be true.

	From the above proof correctness follows: if $\mathcal A$ is an input-consistent $2$-NFA and there exists a truth assignment satisfying the formula, then the assignment induces a total ordering of the nodes satisfying the Wheeler properties. Conversely, the algorithm is clearly complete: if $\mathcal A$ is a Wheeler $2$-NFA, then there exists a total ordering of the nodes satisfying the Wheeler properties. This defines a truth assignment of the variables that satisfies our 2-SAT formula. \qed
\end{proof}

We note that it is tempting to try to generalize the above solution to general NFAs by simulating arbitrary degree-$d$ nondeterminism  using binary trees: a node with $d$ equally-labeled outgoing edges could be expanded to a binary tree with $d$ leaves (bringing down the degree of nondeterminism to $2$). Unfortunately, while this solution works for transitivity (which is successfully propagated from the source), it could make the graph non-Wheeler: the  topology of those trees cannot be arbitrary and must satisfy the co-lexicographic ordering of the nodes, i.e. the solution we are trying to compute.

\subsection{Sorting WDFAs Online}\label{app: Sorting WDAGs Online}

Algorithm \ref{alg:sort} initializes all variables used by our procedure and implements Kahn's topological-sorting algorithm~\cite{kahn1962topological}. Every time a new node is appended to the topological ordering, we call Algorithm \ref{alg:step}---our actual online algorithm---to update also the co-lexicographic ordering. 
This step also checks if the new node and its incoming edges falsify the Wheeler properties. 
We use the following structures (indices start from 1):

\begin{itemize}
	\item $\mathtt{LEX}$ is a dynamic sequence of distinct nodes $v_1, \dots, v_k\in V$ supporting the following operations: 
	\begin{enumerate}
		\item $\mathtt{LEX[i]}$ returns $v_i$.
		\item $\mathtt{LEX^{-1}[v]}$, with $v\in \mathtt{LEX}$, returns the index $i$ such that $\mathtt{LEX}[i]=v$.
		\item $\mathtt{LEX.insert(v,i)}$ inserts node $v$ between $\mathtt{LEX[i-1]}$ and $\mathtt{LEX[i]}$. If $i=1$, $v$ is inserted at the beginning of the sequence. This operation increases the sequence's length by one. \\
	\end{enumerate}
	\item $\mathtt{IN}$ and $\mathtt{OUT}$ are dynamic sequences of strings, i.e. sequences $\alpha_1, \dots, \alpha_k$, where $\alpha_i\in\Sigma^*$ (note that $\alpha_i$ could be the empty string $\epsilon$). 
	%At each computation step, we maintain $\mathtt{IN}, \mathtt{OUT}$, and $\mathtt{LEX}$ synchronized as follows This allows us to 
	To make our pseudocode more readable, we index $\mathtt{IN}$ and $\mathtt{OUT}$ by nodes of $\mathtt{LEX}$ (these three arrays will be synchronized). Let $\mathtt{T} = \alpha_1, \alpha_2,\dots, \alpha_k$, with $\mathtt{T\in \{IN,OUT\}}$. Both arrays support the following operation:
	\begin{enumerate}
		\setcounter{enumi}{3}
		\item $\mathtt{T.insert(\alpha,v)}$, where $\alpha\in\Sigma^*$ and $v\in \mathtt{LEX}$: insert $\alpha$ between $\alpha_{\mathtt{LEX}^{-1}[v]-1}$ and $\alpha_{\mathtt{LEX}^{-1}[v]}$. If $\mathtt{LEX^{-1}[v]}=1$, then $\alpha$ is inserted at the beginning of $\mathtt{T}$. This operation increases the sequence's length by one. 
	\end{enumerate}
	\ \\
	Sequence $\mathtt{OUT}$ supports these additional operations:
	
	\begin{enumerate}	
		\setcounter{enumi}{4}
		\item $\mathtt{OUT}[v]$, with $v\in \mathtt{LEX}$, returns $\alpha_{\mathtt{LEX}^{-1}[v]}$.
		\item $\mathtt{OUT.append(\alpha,v)}$, where $\alpha\in\Sigma^*$ and $v\in \mathtt{LEX}$: append the string $\alpha$ at the end of the string $\mathtt{OUT[v]}$, i.e. replace $\mathtt{OUT[v]} \leftarrow \mathtt{OUT[v]}\cdot \alpha$. Note that this operation does not increase $\mathtt{OUT}$'s length. 
		\item $\mathtt{OUT.rank(c,u)}$, with $u\in \mathtt{LEX}$ and $c\in \Sigma$: return the number of characters equal to $c$ in all strings $\mathtt{OUT[v]}$, with $v = \mathtt{LEX[1]}, \mathtt{LEX[2]}, \dots, \mathtt{LEX[LEX^{-1}[u]]}$.
		\item $\mathtt{OUT.reserve(u,v,c)}$, with $u,v\in \mathtt{LEX}$ and $c\in \Sigma$: from the moment this operation is called, the sequence $\alpha_{\mathtt{LEX^{-1}[u]}}, \dots, \alpha_{\mathtt{LEX^{-1}[v]}}$ is marked with label $c$. Note that inserting new elements inside $\alpha_{\mathtt{LEX^{-1}[u]}}, \dots, \alpha_{\mathtt{LEX^{-1}[v]}}$ will increase the length of the reserved sequence.
		\item  $\mathtt{OUT.is\_reserved(v,c)}$,  with $v\in \mathtt{LEX}$ and $c\in \Sigma$: return $\mathtt{TRUE}$ iff $\alpha_{\mathtt{LEX}^{-1}(v)}$ falls inside a sequence that has been marked (reserved) with character $c$.
	\end{enumerate}
	
	In our algorithm, sequence $\mathtt{IN}$ will always be partitioned in at most $t\leq \sigma+1$ sub-sequences $\mathtt{IN} = \alpha^{c_1}_1, \dots, \alpha^{c_1}_{k_{c_1}}, \alpha^{c_2}_1, \dots, \alpha^{c_2}_{k_{c_2}}, \dots, \alpha^{c_t}_1, \dots, \alpha^{c_t}_{k_{c_t}}$, where each $\alpha^{c}_i$ contains only character $c$ and $c_1 \prec c_2 \prec \dots \prec c_t$. We define an additional operation on $\mathtt{IN}$:\\
	\begin{enumerate}	
		\setcounter{enumi}{9}
		\item  $\mathtt{IN.start(c)}$, with $c\in \Sigma$, returns the largest integer $j\geq 1$ such that all characters in $\mathtt{IN}[v]$ are strictly smaller than $c$, for all $v = \mathtt{LEX}[1], \dots, \mathtt{LEX}[j-1]$. 
	\end{enumerate}
	
\end{itemize}

Figure \ref{fig:inconsistencies} shows how our dynamic structures evolve while processing states in topological order.
In Appendix \ref{sec:data structures} we discuss data structures implementing the above operations in $O(\log k)$ time, $k$ being the sequence's length. 
Intuitively, these three dynamic sequences have the following meaning: $\mathtt{LEX}$ will contain the co-lexicographically-ordered sequence of nodes. $\mathtt{IN}[v]$ and $\mathtt{OUT}[v]$, with $v\in\mathtt{LEX}$, will contain the labels of the incoming and outgoing edges of $v$, respectively. To keep the three sequences synchronized, when inserting $v$ in $\mathtt{LEX}$ we will also need to update the other two sequences so that $\mathtt{IN}[v] = c^t$, where $t$ is the number of incoming edges, labeled $c$, of $v$, and $\mathtt{OUT}[v] = \epsilon$, since $v$ does not have yet outgoing edges. $\mathtt{OUT}[v]$ will (possibly) be updated later, when new nodes adjacent to $v$ will arrive in the topological order. 
Our representation is equivalent to that used in~\cite{siren2014indexing} to represent the GCSA data structure. Intuitively, $\mathtt{OUT}$ is a generalized version of the well-known Burrows-Wheeler transform (except that we sort prefixes in co-lexicographic order instead of suffixes in lexicographic order). If the graph is a path (i.e. a string) then $\mathtt{OUT}$ is precisely the BWT of the reversed path. 

We proceed with a discussion of the pseudocode. In Lines \ref{line:init in}-\ref{line:init OUT} of Algorithm \ref{alg:sort} we initialize all variables and data structures. Let $u\in V$. The variable $\mathtt{u.in}$ memorizes the number of incoming edges in $u$; we will use this counter to implement Kahn's topological sorting procedure. $\mathtt{u.label}$ is the label of all incoming edges of $u$, or $\#$ if $u=s$. $\mathtt{IN}, \mathtt{LEX}$, and $\mathtt{OUT}$ are initialized as empty dynamic sequences. 
Lines \ref{line: init S}-\ref{line:cycle} implement Kahn's topological sorting algorithm~\cite{kahn1962topological}. Each time a new node $u$ is appended to the order, we call our online procedure $\mathtt{update(u)}$, implemented in Algorithm \ref{alg:step}. Algorithm \ref{alg:step} works as follows. Assume that we have already sorted $v_1, \dots, v_k$, that $\mathtt{LEX}$ contains the nodes' permutation reflecting their co-lexicographic order, and that $\mathtt{IN[v_i]}$ and $\mathtt{OUT[v_i]}$ contain the incoming and outgoing labels for each $i=1, \dots, k$ in the sub-graph induced by $v_1, \dots, v_k$. When a new node $u$ arrives in topological order, all its $t$ predecessors are in $\mathtt{LEX}$. 
Let $b = \mathtt{u.label}$ be the incoming label of $u$.
We find the co-lexicographically smallest $v_{min}$ and largest $v_{max}$ predecessors of $u$ (using function $\mathtt{LEX^{-1}}$ on all $u$'s predecessors). In our pseudocode, if $u=s$ then $v_{min} = v_{max} = \mathtt{NULL}$.
To keep the Wheeler properties true, note that
there cannot be $b$'s in the range $\mathtt{OUT[v_{min}..v_{max}]}$: if there are, since we will append $b$ to $\mathtt{OUT[v_{min}]}$ and $\mathtt{OUT[v_{max}]}$, there will be three nodes $v_{min} < v' < v_{max}$ such that $(v_{min},u,b), (v',u',b), (v_{max},u, b) \in E$ for some $u'$. Then, by Wheeler property (ii), this would imply that $u < u' < u$, a contradiction. We therefore check this event using function $\mathtt{contains}$ (note: this function can be easily implemented using two calls to $\mathtt{rank}$). If $b$'s are present, then the graph is no longer Wheeler: such an event is shown in Figure \ref{fig:inconsistencies}, left-hand side (where $u = v_5$). Otherwise, the number $j$ of $b$'s before $v_{min}$ (which is equal to the number of $b$'s before $v_{max}$) tells us the co-lexicographic rank $i$ of $u$ (similarly to the standard string-BWT, we obtain this number by adding $j$ to the starting position of $b$'s in $\mathtt{IN}$), and we can mark (reserve) range $\mathtt{OUT[v_{min}..v_{max}]}$ with letter $b$ using function $\mathtt{reserve}$. 
%Note that, by the way we compute $i$, conditions {\bf 1} and {\bf 3} of Definition \ref{def:range-consistent} are always satisfied. 
Such an event is shown in Figure \ref{fig:inconsistencies}, left-hand side, when inserting, e.g., node $v_3$. 
%If the consistency check succeeds, 
At this point, we may have an additional inconsistency falsifying the Wheeler properties in the case that one of the predecessors $v_i$ of $u$ falls inside a reserved range for $b$ (reserved by a node other than $u$): this happens, for example, when inserting $v_6$ in Figure \ref{fig:inconsistencies}, right-hand side. This check requires calling function $\mathtt{is\_reserved}$.
If all tests succeed, we 
insert $u$ in position $i$ of $\mathtt{LEX}$ and we update $\mathtt{IN}$ and $\mathtt{OUT}$ by inserting $b^t$ at the $i$-th position in $\mathtt{IN}$ (i.e. the position corresponding to $u$) and by appending $b$ at the end of each $\mathtt{OUT[v_i]}$ for each predecessor $v_i$ of $u$.

\begin{algorithm}[th!]
	\caption{\texttt{sort(G)}}
	\label{alg:sort}
	
	\SetKwInOut{Input}{input}
	\SetKwInOut{Output}{output}
	\SetSideCommentLeft
	\LinesNumbered
	
	\Input{Labeled DAG $G = (V,E,s,\Sigma)$}
	\Output{A permutation of $V$ reflecting the co-lexicographic ordering of the nodes, or $\mathtt{FAIL}$ if such an ordering does not exist.}
	\BlankLine
	\BlankLine
	
	\BlankLine
	
	\For{$u \in V$}{
		
		\BlankLine
		
		$\mathtt{u.in \leftarrow 0}$\;\label{line:init in}
		$\mathtt{u.label \leftarrow \mathtt{NULL}}$\; 
		
		\BlankLine
		
	}
	
	\BlankLine	
	$\mathtt{s.label \leftarrow \#}$;
	\BlankLine
	
	\For{$(u,v,a) \in E$}{
		
		\BlankLine
		
		$\mathtt{v.in \leftarrow v.in +1}$\;
		
		\BlankLine
		
		\If{$\mathtt{v.label} \neq \mathtt{NULL}$\ $\mathbf{and}$\  $\mathtt{v.label} \neq a$}{
			
			\Return \texttt{FAIL}\tcc*[r]{Cannot be Wheeler graph}	
			
		}
		
		$\mathtt{v.label} \leftarrow a$\;
	}
	
	\BlankLine
	
	$\mathtt{IN} \leftarrow \mathtt{new\_dyn\_sequence(\Sigma^*)}$\tcc*[r]{Sequence of strings}
	$\mathtt{LEX} \leftarrow \mathtt{new\_dyn\_sequence(V)}$\tcc*[r]{Sequence of nodes}
	$\mathtt{OUT} \leftarrow \mathtt{new\_dyn\_sequence(\Sigma^*)}$\tcc*[r]{Sequence of strings}\label{line:init OUT}
	
	\BlankLine
	
	$S \leftarrow \{s\}$\tcc*[r]{Set of nodes with no incoming edges}\label{line: init S}
	
	\BlankLine	
	
	\While{$S\neq \emptyset$}{
		
		\BlankLine
		$\mathtt{u \leftarrow S.pop()}$\tcc*[r]{Extract any $u\in S$}
		$\mathtt{update(u)}$\tcc*[r]{Call to Algorithm \ref{alg:step}. If this fails, return \texttt{FAIL}.}
		
		\BlankLine
		\For{$(u,v,a) \in E$}{
			\BlankLine	
			
			$v.in \leftarrow v.in-1$\;
			
			\If{$v.in = 0$}{
				
				$S \leftarrow S \cup \{v\}$\;	
				
			}
			
			\BlankLine
		}
		
	}
	
	\BlankLine
	
	\If{$\exists\ v\in V\ :\ v.in>0$}{
		
		\Return $\mathtt{FAIL}$\tcc*[r]{cycle found!}\label{line:cycle}
		
	}
	
	\BlankLine
	
	\Return $\mathtt{LEX}$\;
	
\end{algorithm}

\begin{algorithm}[th!]
	\caption{\texttt{update(u)}}
	\label{alg:step}
	
	\SetKwInOut{Input}{input}
	\SetKwInOut{Output}{behavior}
	\SetSideCommentLeft
	\LinesNumbered
	
	\Input{Node $u$}
	\Output{Inserts $u$ at the right place in the co-lexicographic ordering $\mathtt{LEX}$ of the nodes already processed, or returns $\mathtt{FAIL}$ if a conflict is detected.}
	\BlankLine
	\BlankLine
	
	$\mathtt{v_{min} \leftarrow min\_pred(u)}$\tcc*[r]{co-lexicographically-smallest predecessor}
	$\mathtt{v_{max} \leftarrow max\_pred(u)}$\tcc*[r]{co-lexicographically-largest predecessor}
	
	\BlankLine

	\eIf{$u \neq s$}{
		
		\BlankLine
		
		\eIf{$\mathtt{OUT[v_{min},\dots,v_{max}].contains(u.label)}$}{
			
			\BlankLine	
			\Return \texttt{FAIL}\tcc*[r]{Inconsistency of type 1}
			\BlankLine
			
		}{
			
			\BlankLine
			\For{$\mathtt{(v,u,a)\in E}$}{
				\BlankLine
				\eIf{$\mathtt{OUT.is\_reserved(v,u.label)}$}{
					
					\Return \texttt{FAIL}\tcc*[r]{Inconsistency of type 2}
					
				}{
					
					$\mathtt{OUT[v].append(u.label)}$\;
					
				}
				\BlankLine
				
			}
			
			\BlankLine
			$\mathtt{OUT.reserve(v_{min},v_{max},u.label)}$\tcc*[r]{Reserve $\mathtt{[v_{min},v_{max}]}$ with $\mathtt{u.label}$}
			\BlankLine
			
		}
		
		\BlankLine
		$i \leftarrow \mathtt{IN.start(u.label) + OUT.rank(u.label,v_{min})}$\;
		$\mathtt{LEX.insert(u,i)}$\;
		$p \leftarrow \mathtt{|pred(u)|}$\tcc*[r]{Number of predecessors of $u$}
		\BlankLine
		
	}{
		\BlankLine
		$\mathtt{LEX.insert(u,1)}$\;
		$p \leftarrow 1$\tcc*[r]{Number of predecessors of $u$}
		\BlankLine
		
	}
	
	\BlankLine
	
	$\mathtt{IN.insert(u.label}^p,u)$\tcc*[r]{Insert $p$ times $\mathtt{u.label}$}
	$\mathtt{OUT.insert(\epsilon,u)}$\tcc*[r]{$u$ does not have successors yet}

\end{algorithm}

\begin{figure}
	\begin{center}
		\includegraphics[scale=0.41]{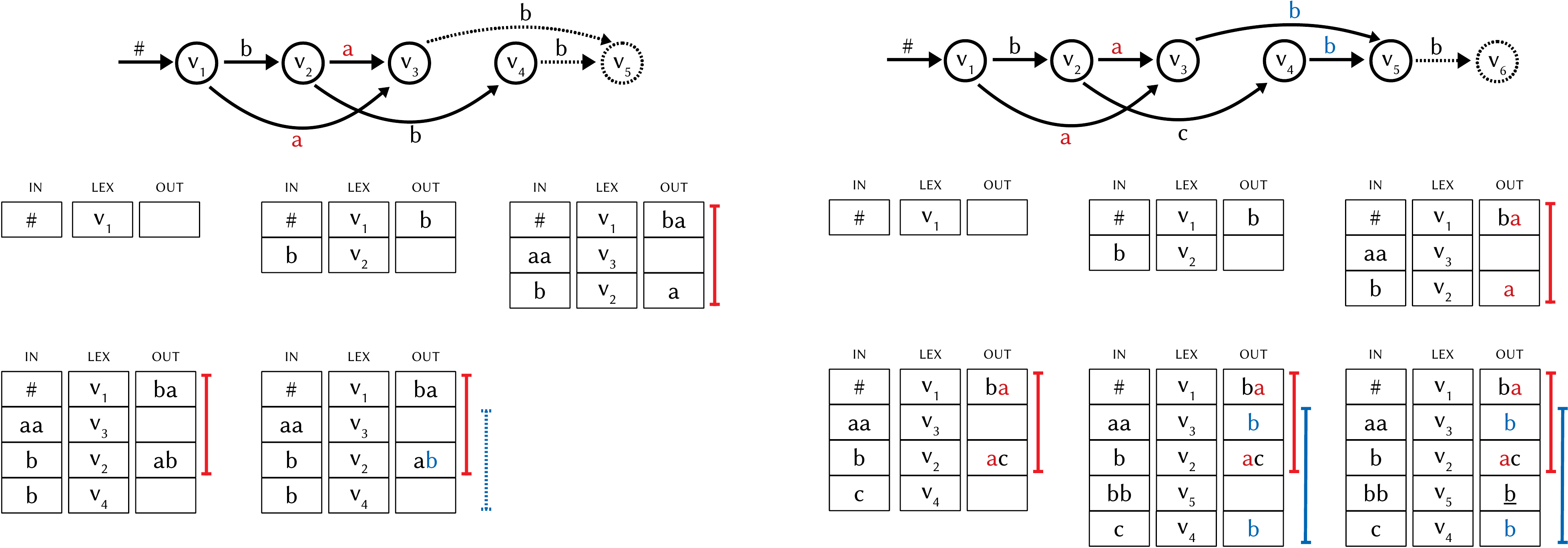}
	\end{center}
	\caption{\textbf{Left}: Inconsistency of type 1. The five tables show how arrays $\mathtt{IN}$, $\mathtt{LEX}$, and $\mathtt{OUT}$ evolve during insertions of nodes $v_1, \dots, v_5$ in topological order. Up to node $v_4$, the Directed Acyclic Graph (DAG) is Wheeler. 
		When inserting node $v_3$, we successfully reserve the interval $[v_1,v_2]$ with label 'a' (shown in red). From this point, no 'a's can be inserted inside the reserved interval.
		When inserting node $v_5$ (with incoming label 'b'), the co-lexicographically smallest and largest predecessors of $v_5$ are $v_3$ and $v_4$, respectively. This means we have to reserve the interval $[v_3,v_4]$ with label 'b' (shown in blue dashed line); however, this is not possible since there already is a 'b' (highlighted in blue) in $\mathtt{OUT}[v_3,\dots,v_4]$.
		\textbf{Right}: Inconsistency of type 2. Up to node $v_5$, the DAG is Wheeler. 
		Note that we successfully reserve two intervals: $[v_1,v_2]$ (with letter 'a', red interval), and $[v_3,v_4]$ (with letter 'b', blue interval).
		When inserting node $v_4$ we do not need to reserve any interval since the node has only one predecessor. 
		The confict arises when inserting node $v_6$ (with incoming label 'b'). Since $v_5$ is a predecessor of $v_6$, we need to append 'b' in $\mathtt{OUT[v_5]}$. However, this 'b' (underlined in the picture) falls inside a reserved interval for 'b' (in blue).
}\label{fig:inconsistencies}
\end{figure}

\subsection{Data Structure Details}\label{sec:data structures}

In this section we show how to implement operations {\bf 1}-{\bf 10} used by Algorithms \ref{alg:sort} and \ref{alg:step} using state-of-the-art data structures. At the core of $\mathtt{LEX}$ and $\mathtt{OUT}$ stands the dynamic sequence representation of Navarro and Nekrich~\cite{navarro2014optimal}. This structure supports insertions, access, rank, and select in $O(\log n)$ worst-case time, $n$ being the sequence's length. The space usage is bounded by $nH_0 + o(n\log\sigma) + O(\sigma\log n)$ bits, where $H_0$ is the zero-th order entropy of the sequence. Sequence $\mathtt{IN}$ will instead be represented using a dynamic partial sum data structure, e.g. a balanced binary tree or a Fenwick tree~\cite{fenwick1994new}, and a dynamic bitvector. All details follow. 

Sequence $\mathtt{LEX}$ is stored with Navarro and Nekrich's dynamic sequence representation~\cite{navarro2014optimal}. Operations {\bf 1}-{\bf 3} are directly supported on the representation. Operation {\bf 2} is simply a $select_{v}(1)$ (i.e. the position of the first $v$).

Sequence $\mathtt{OUT}$ is stored using a dynamic sequence $\mathtt{out}$ and a bitvector, both represented with Navarro and Nekrich's dynamic sequence. The idea is to store all the strings  $\mathtt{OUT}[1], \dots, \mathtt{OUT}[|\mathtt{OUT}|]$ concatenated in a single sequence $\mathtt{out}$, and mark the beginning of those strings with a bit set in a dynamic bitvector $\mathtt{B_{out}[1..n]}$, were $n=|\mathtt{out}|$.
Clearly, operations {\bf 4}-{\bf 7} on $\mathtt{OUT}$ can be simulated with a constant number of operations (\emph{insert, access, rank, select}) on $\mathtt{out}$ and $\mathtt{B_{out}}$. 

Operations {\bf 8}-{\bf 9} require an additional dynamic sequence of parentheses $\mathtt{PAR}[1..n]$ on alphabet $\{ \mathtt{(}_c\ :\ c\in\Sigma\} \cup \{ \mathtt{)}_c\ :\ c\in\Sigma\} \cup \{ \square \}$. Every time a new character is inserted at position $i$ in $\mathtt{out}$, we also insert $\square$ at position $i$ in $\mathtt{PAR}$. When $\mathtt{OUT.reserve(u,v,c)}$ is called (i.e. operation {\bf 8}), let $i_v$ and $i_u$ be the positions in $\mathtt{out}$ corresponding to the two occurrences of character $c$ in $\mathtt{OUT[u]}$ and $\mathtt{OUT[u]}$, respectively (remember that the automaton is deterministic, so these positions are unique). These positions can easily be computed in $O(\log n)$ time using \emph{select} and \emph{rank} operations on $\mathtt{out}$ and $\mathtt{B_{out}}$. Then, we replace $\mathtt{PAR}[i_u]$ and $\mathtt{PAR}[i_v]$ with characters $\mathtt{(_c}$ and $\mathtt{)_c}$, respectively (replacing a character requires a deletion followed by an insertion). 
Note that reserved intervals for a fixed character do not overlap, so this parentheses representation permits to unambiguously reconstruct the structure of the intervals. 
At this point, operation {\bf 9} is implemented as follows. Let $i_v$ be the position in $\mathtt{out}$ corresponding to the first character in $\mathtt{OUT}[v]$. 
This position can be computed in $O(\log n)$ time with two \emph{select} operations on $\mathtt{LEX}$ and $\mathtt{B_{out}}$.
Then, $\mathtt{OUT.is\_reserved(v,c)}$ returns true if and only if $\mathtt{PAR.rank_{(_c}(i_v)}>\mathtt{PAR.rank_{)_c}(i_v)}$, i.e. if we did not close all opening parentheses $\mathtt{{(_c}}$ before position $i_v$ (note that does not make any difference if $i_v$ is the first or last position in $\mathtt{OUT}[v]$, since when we call this operation $\mathtt{OUT}[v]$ does not contain  characters equal to $c$).

To conclude, $\mathtt{IN}$ is represented with a dynamic bitvector $\mathtt{B_{IN}[1..n]}$ and a partial sum $\mathtt{PS[1..\sigma+1]}$ supporting the following operations in $O(\log\sigma)$ time:

\begin{itemize}
	\item \emph{partial sum}: $\mathtt{PS.ps(i)} = \sum_{j=1}^i\mathtt{PS[j]}$.
	\item \emph{update}: $\mathtt{PS[i]} \leftarrow \mathtt{PS[i]}+\delta$.
\end{itemize}

Fenwick trees~\cite{fenwick1994new} support the above operations within this time bound. 
Bitvector $\mathtt{B_{IN}[1..n]}$ contains the bit sequence $110^{t_{v_2}-1} 10^{t_{v3}-1} \dots 10^{t_{v_k}-1}$, where $t_{v_i}$ is the number of predecessors of $v_i$ in the current sequence $\mathtt{LEX} = v_1, \dots, v_k$ of sorted nodes (note that $v_1$ is always the source $s$). 
Assume, for simplicity, that $\Sigma = [1,\sigma+1]$, where $\#$ corresponds to $1$ (this is not restrictive, as we can map the alphabet to this range at the beginning of the computation).
At the beginning, the partial sum is initialized as $\mathtt{IN[c]=0}$ for all $c$. Operation {\bf 10}, $\mathtt{IN.start(c)}$, with $c\neq \#$, is implemented as $\mathtt{PS.ps(c-1)+1}$. If $c=\#$, the operation returns $1$. Operation {\bf 4}, $\mathtt{IN.insert(c^p,u)}$, is implemented as $\mathtt{PS[c]} \leftarrow \mathtt{PS[c]}+p$ followed by  $\mathtt{B_{IN}.insert(10^{p-1},i_u)}$ (i.e. $p$ calls to \emph{insert} on the dynamic bitvector at position $i_u$), where $i_u$ is the position of the $j$-th bit set in $\mathtt{B_{IN}}$ (a \emph{select} operation) and $j = \mathtt{LEX^{-1}[u]}$, or $i_u=n+1$ if $\mathtt{B_{IN}}$ has $j-1$ bits set (note that, when we call $\mathtt{IN.insert(c^p,u)}$, node $u$ has already been inserted in $\mathtt{LEX}$).

\subsection{Proof of Theorem \ref{thm:n log n}}\label{proof thm n log n}

In Appendix \ref{sec:data structures} we show that all operations on the dynamic sequences can be implemented in logarithmic time. Correctness follows from the fact that we always check that the Wheeler properties are maintained true. 
To prove completeness, note that at each step we place $u$ between two nodes $v_1$ and $v_2$ in array $\mathtt{LEX}$ only if the smallest $u$'s predecessor is larger than the largest $v_1$'s predecessor, and if the largest $u$'s predecessor is smaller than the smallest $v_2$'s predecessor. This is the only possible choice we can make in order to satisfy $w_{v_1} \prec w_{u} \prec w_{v_2}$ for all strings labeling paths $s \rightsquigarrow v_1$, $s \rightsquigarrow u$, and $s \rightsquigarrow v_2$ and to obtain, by Corollary \ref{lem:clusters}, the only possible correct ordering of the nodes. It follows that, if the new node $v$ does not falsify the Wheeler properties, then we are computing its co-lexicographic rank correctly.

\begin{proof} (of Lemma \ref{lem: check range consistency}) \label{proof lemma check range consistency}

First, we sort edges by label, with ties broken by origin, and further ties broken by destination. This can be achieved in time $O(|E| + |V|)$ by radix sorting the edges represented as triples $(a, u, v)$, where $a$ is the label, and $u$ and $v$ respectively are the ranks of the source and destination nodes in the given order $<$. 

Let $L$ denote the sorted list of edges. We claim that the given order $<$ satisfies the Wheeler properties (Definition \ref{def_WG}) if and only if for all pairs of consecutive edges $(a_i, u_i, v_i), (a_{i+1}, u_{i+1}, v_{i+1})$ in $L$, we have $(a_i = a_{i+1}) \rightarrow v_i \leq v_{i+1}$ and $(a_i \neq a_{i+1}) \rightarrow v_i < v_{i+1}$. Clearly this can be checked in time $O(|E|)$ with one scan over $L$. We now argue the correctness of this algorithm.

Wheeler property (ii) is equivalent to the condition that when all edges labeled by some character $a \in \Sigma$ are sorted by source with ties broken by destination, the sequence of destinations is monotonically increasing, which is expressed by the condition $(a_i = a_{i+1}) \rightarrow v_i \leq v_{i+1}$.

Wheeler property (i) is equivalent to the condition that for all pairs of characters $a,b \in \Sigma$ such that $b$ is a successor of $a$ in the order of $\Sigma$, denoting by $v_a$ the largest node with an incoming $a$-edge, and by $v_b$ the smallest node with an incoming $b$-edge, we have $v_a < v_b$. If Wheeler property (ii) holds, then destinations $v_a$ and $v_b$ are consecutive in $L$ because the list is sorted primarily by label and destinations are monotonically increasing for each label. Hence checking for $(a_i \neq a_{i+1}) \rightarrow v_i < v_{i+1}$ verifies Wheeler property (i) given that Wheeler property (ii) holds. \qed

\end{proof}

\begin{proof} (of Theorem \ref{thm:DFA linear}) \label{proof thm: DFA linear}

	In $O(|V|+|E|)$ time we build a directed spanning tree $\mathcal T$ of $\mathcal A$ with root $s$ (e.g. its directed shortest-path tree with root $s$). Note that this is always possible since we assume that all states are reachable from $s$. 

	By Corollary \ref{lem:clusters}, if $\mathcal A$ is a Wheeler graph then we can use the strings that label \emph{any} two paths $s\rightsquigarrow u$ and $s\rightsquigarrow v$ to decide the order of any two nodes $u$ and $v$. We can therefore sort $V$ according to the paths spelled by $\mathcal T$; by Corollary \ref{lem:clusters}, if $\mathcal A$ is Wheeler then we obtain the correct (unique) ordering. To prefix-sort $\mathcal T$, we compute its XBW transform~\footnote{note: this requires mapping the labels of $\mathcal T$ to alphabet $\Sigma' \subseteq [1,|V|]$ while preserving their lexicographic ordering. Since we assume that the original alphabet's size does not exceed $|E|^{O(1)} = |V|^{O(1)}$, this step can be performed in linear time by radix-sorting the labels.} in $O(|V|)$ time~\cite[Thm 2]{ferragina2009compressing}. The array containing the lexicographically-sorted nodes (i.e. the prefix array of $\mathcal T$) can easily be obtained from the XBW transform using, e.g. the partial rank counters defined in the proof of Lemma \ref{lem: check range consistency} to navigate the tree (this is analogous to repeatedly applying function LF on the BWT in order to obtain the suffix array). At this point, we check that the resulting node order satisfies the Wheeler properties using Lemma \ref{lem: check range consistency}. If this is this case, then the above-computed prefix array contains the prefix-sorted nodes of $\mathcal A$. \qed
\end{proof}

\section{Proofs of Section \ref{sec:minimization} }

\begin{proof} (of Theorem \ref{thm: min DFA})\label{proof thm: min DFA}

	Let $\mathcal A=(V,E,F,s,\Sigma)$. Consider the (possibly infinite) deterministic automaton $\mathcal T$ that is a tree and that is equivalent to $\mathcal A$ in the following sense: $\mathcal T$ is the (unique) tree obtained by ``unraveling'' $\mathcal A$, i.e. the tree containing all words in $\mathcal L(\mathcal A)$ such that each path labeled with such a word leads to an accepting state. 	 
	%The initial state of $\mathcal T$ is its root. 
	%We moreover mark each state $u^j\in L^u$ of $\mathcal T$ as final if and only if $u$ is final in $\mathcal A$. 
	Clearly, $\mathcal T$ is a (possibly infinite) deterministic automaton recognizing $\mathcal L(\mathcal A)$: a string $\alpha$ leads to a final state in $\mathcal A$ if and only if it does in $\mathcal T$. 
	
	Let $L^u = \{u^1, u^2, \dots, u^{k_u}\}$ be the (possibly infinite) set of nodes of $\mathcal T$ reached by following, from its root, all the paths labeled $\alpha$ for each $\alpha$ labeling a path $s \rightsquigarrow u$ connecting $s$ with $u$ in $\mathcal A$. 
	Note that each state $u$ of $\mathcal A$ can be identified by the set $L^u$ of states of $\mathcal T$; this allows us to extend $\equiv_w$ to the states of $\mathcal T$ as follows: $u^i \equiv_w u^j$ for all $u^i,u^j\in L^u$, $u\in V$, and $u^i \equiv_w v^j$ for  $u^i\in L^u$, $v^j\in L^v$ if and only if $u \equiv_w v$.
	
	Consider now the process of minimizing $\mathcal T$ by collapsing states in equivalence classes in such a way that
	(i) the quotient automaton is finite, (ii) the accepting language of the quotient DFA is the same as that of $\mathcal T$ and  (iii) the quotient DFA is Wheeler. 
	By the existence of $\mathcal A$, there exists such a partition (not necessarily the coarsest): the one putting $u^i$ and $u^j$ in the same equivalence class if and only if $u^i, u^j \in L^u$, for some $u\in V$ (in this case, $\mathcal A$ itself is the resulting quotient automaton).
	Call $\equiv$ the relation among states of  $\mathcal T$ yielding the \emph{smallest} such WDFA $\mathcal A/_\equiv$. By definition, $\mathcal A/_\equiv$ is the smallest WDFA  recognizing $\mathcal L(\mathcal A)$. Our claim is that $\equiv\ =\ \equiv_w$, i.e. that Algorithm \ref{alg:minimize} returns this automaton.

	We observe that: 
	
	\begin{enumerate}
		\item  $u^i \approx u^j$ for any $u^i,u^j\in L^u$ and all $u\in V$. Otherwise, assume for a contradiction that there exists a string $\alpha$ leading to an accepting state from $u^i$ but not from $u^j$. By construction of $\mathcal T$, $u^i$ and $u^j$ are $\approx$-equivalent to $u$: this leads to a contradiction, since the state reached from $u$ with label $\alpha$ cannot be both accepting and not accepting.
		\item Since $\mathcal A$ is a Wheeler DFA, Corollary \ref{lem:clusters} applied to $\mathcal A$ tells us that, for any two nodes $u<v\in V$, all strings labeling paths from the root of $\mathcal T$ to nodes in $L^u$ are co-lexicographically smaller than those labeling paths from the root of $\mathcal T$ to nodes in  $L^v$. We express this fact using the notation $L^u < L^v$.
		\item Since $\mathcal A$ is Wheeler, then each $u\in V$ has only one distinct incoming label and $\lambda(u^j) = \lambda(u)$ for all $u^j\in L^u$.
	\end{enumerate}
	
	By the above properties, $u^i \equiv u^j$ for all $u^i,u^j\in L^u$, $u\in V$. To see this, note that, by property {\bf 1}, those states are all equivalent by relation $\approx$. Moreover, properties {\bf 2-3} combined with Corollary \ref{lem:clusters} imply that, by grouping states in each  $L^u$, we cannot break any Wheeler property. It follows that $\equiv$ must group those states, being the coarsest partition finer than $\approx$ with these two properties. Let us indicate with $L^u \equiv L^v$ the fact that $u^i \equiv v^j$ for all $u^i\in L^u,\ v^j\in L^v$.
	
	Suppose now, for a contradiction, that there exist $L^u < L^v < L^w$ with $L^u \equiv L^w \not\equiv L^v$. Then, by Corollary \ref{lem:clusters},  $L^u < L^v$ implies that, in the quotient automaton, states $[L^w]_\equiv = [L^u]_\equiv$ and $[L^v]_\equiv$ are reachable from the source by two paths $\alpha$ and $\beta$, respectively, with $\alpha \prec \beta$. Conversely, $L^v < L^w$ implies that states $[L^v]_\equiv$ and $[L^w]_\equiv$ are reachable from the source by two paths $\alpha'$ and $\beta'$, respectively, with $\alpha' \prec \beta'$. Then, by Corollary \ref{lem:clusters} we cannot define a total order on $\mathcal A/_\equiv$'s states, i.e. $\mathcal A/_\equiv$ is not Wheeler. 
	
	By all the above observations, we conclude that $\equiv$ must (i) group only equivalent states by $\approx$, (ii) group only states with the same incoming label, (iii) group all states inside each $L^u$, and (iv) group only states in \emph{adjacent} sets $L^u$, $L^v$ in the co-lexicographic order. By its definition, the relation $\equiv_w$ induces the coarsest partition that satisfies (i)-(iv), therefore we conclude that $\equiv\ =\ \equiv_w$.	\qed
\end{proof}

\subsection{Converting DFAs to minimum WDFAs}\label{app:DFA->WDFA}

We describe an online step of our algorithm. Assume we successfully built $\mathcal A_i$, with $i<t$, and we are about to process $v_{i+1}$ in order to build $\mathcal A_{i+1}$. Let $\{c_1, \dots, c_k\}$ be the labels of incoming $v_{i+1}$'s edges. We first replace (split) $v_{i+1}$ by $k$ equivalent states $v_{i+1}^{c_1} \approx \dots \approx v_{i+1}^{c_k}$: each $v_{i+1}^{c_k}$ (i) is accepting if and only if $v_{i+1}$ is accepting, (ii) keeps only the incoming edges of $v_{i+1}$ labeled $c_i$, and (iii) it duplicates all its outgoing edges: we replace each $(v_{i+1},u,c)$ with the edges $(v_{i+1}^{c_1},u,c), \dots, (v_{i+1}^{c_k},u,c)$. Note that all the newly-created edges must be present in the final automaton $\mathcal A_t$ since the states $v_{i+1}^{c_1}, \dots, v_{i+1}^{c_k}$ cannot be collapsed back by $\equiv_w$ (as they have different incoming labels); it follows that in this step we are not creating more edges than necessary.

We now insert separately $v_{i+1}^{c_1}, \dots, v_{i+1}^{c_k}$ in $\mathtt{LEX}_i$ in any order as follows. The procedure is the same for all those vertices, therefore we may simply assume we are about to process a node $v$ with all incoming edges labeled with the same character $a$. Let $u_1 < \dots < u_k$ be the predecessors of $v$ in the graph; note that those nodes must belong to $\mathtt{LEX}_i$ (since we are processing states in topological order), therefore their order $<$ is well-defined. We now must detect and solve inconsistencies of type 1 and 2 as defined in the proof of Theorem \ref{thm:n log n} (see also Figure \ref{fig:inconsistencies}). 

We start with inconsistencies of type 1: there already are nodes $w_i \notin \{u_1, \dots, u_k\}$ with outgoing edges labeled $a$ inside the range $[u_1, u_k]$.
This breaks the sequence $u_1 < \dots < u_k$ into $q$ sub-intervals $[u_{i_j}, u'_{i_j}]$, $j=1, \dots, q$, that do not contain nodes with outgoing label $a$ different than those in $\{u_1, \dots, u_k\}$. 
The range has therefore the following form, where we denote with $w_i$ and $w'_i$ all nodes not in $\{u_1, \dots, u_k\}$ with outgoing edges labeled $a$ and we highlight in bold the runs $[u_{i_j}, u'_{i_j}]$: 
$$w_1 < \mathbf{u_{i_1} \leq \dots \leq u'_{i_1}} < w_2 \leq \dots \leq w'_2 < \mathbf{u_{i_2} \leq \dots \leq u'_{i_2}} < \dots < \mathbf{u_{i_q} \leq \dots \leq u'_{i_q}} <  w_{q+1} \;,$$
where $u_{i_1} = u_1$, $u'_{i_{q}} = u_k$, and $w_1<u_1,\ w_{q+1}>u_k$ are the rightmost and leftmost states with an outgoing edge labeled $a$, respectively (if they exist).
The top part of Figure \ref{fig:make WDAG} depicts this situation, where $k=4$ and $u_1, \dots, u_4$ are clustered in $q=3$ runs: $w_1 < \mathbf{u_1} < w_2 < w_3 < \mathbf{u_2 < u_3}  < w_4 < \mathbf{u_4} < w_5$.
We solve the inconsistencies of type 1 by splitting $v$ in (i.e. replacing it with) $q$ equivalent nodes: $v_1 \approx \dots \approx v_q$. Each $v_j$ is final if and only if $v$ is final, duplicates all $v$'s outgoing edges (as seen above), and keeps only incoming edges from $v$'s predecessors inside the corresponding run $[u_{i_j}, u'_{i_j}]$. This is depicted in the bottom part of Figure \ref{fig:make WDAG}: $v$ has been split into the three equivalent nodes $v_1 \approx v_2 \approx v_3$.

Inconsistencies of type 2 are solved similarly by splitting $a$-successors of $w_1, \dots, w_{q+1}$ that belong to $\mathtt{LEX}_i$ when necessary. Let $\mathtt{LEX}_i \cap \{succ_a(w_1), \dots, succ_a(w_{q+1})\} = \{z_1 < \dots < z_{q'}\}$ be the $a$-successors of $w_1, \dots, w_{q+1}$ in $\mathtt{LEX}_i$. Note that it might be the case that $q'<q+1$. 
Note also that some of the nodes $z_i$ might belong to $\{u_1, \dots, u_k\} \cup \{w_1, \dots, w_{q+1}\}$.
We have an inconsistency of type 2 (among nodes in $\mathtt{LEX}_i$) whenever $succ_a(w_i) = succ_a(w_{i+1}) = z_e$, for some $1\leq e \leq q'$, and there exist some $u_j$ such that $w_i < u_j < w_{i+1}$. In this case, we split $z_e = succ_a(w_i) =  succ_a(w_{i+1})$ in two equivalent nodes $z_e' \approx z_e''$ ordered as $z_e' < succ_a(u_j) < z_e''$. 
This cannot contradict the Wheeler properties (even if $z_i\in \{u_1, \dots, u_k\} \cup \{w_1, \dots, w_{q+1}\}$), since $succ_a(u_j)$ is one of the copies of $v$ (or $v$ itself if $v$ has not been splitted in the previous step) and has therefore no successors in the current automaton.
The process of fixing inconsistencies of type 2 is shown in Figure \ref{fig:make WDAG}: nodes $w_3$ and $w_4$ are separated by $u_2, u_3$ as $w_3 < u_2 < u_3 < w_4$. In this case, $succ_a(w_3) = succ_a(w_4) = z_3$, and we split $z_3$ in the two equivalent nodes $z'_3$ and $z''_3$. Note also that we only need to check those $w_i$ that immediately precede or follow a predecessor of $v$ (i.e. $w_1, w_2, w_2', \dots , w_{q+1}$): those nodes are at most $O(k)$, where $k$ is the number of $v$'s predecessors. 

As shown in Figure \ref{fig:make WDAG} (bottom), after solving the inconsistencies of type 1 and 2 the nodes in $\mathtt{LEX}_{i+1}$ are again range-consistent: the $a$-successors of any (sorted) range of nodes form themselves a (sorted) range. Moreover, the splitting process defines unambiguously a total ordering of the new nodes among those already in $\mathtt{LEX}_i$, which can be therefore updated to $\mathtt{LEX}_{i+1}$ by inserting those nodes at the right place: to insert a node $v'$ in $\mathtt{LEX}_i$, let $u'$ be its $a$-predecessor: $succ_a(u') = v'$. Let moreover $u'' < u'$ be the rightmost node preceding $u'$ (in $\mathtt{LEX}_i$) having an outgoing edge labeled $a$, and let $v''$ be its $a$-successor: $succ_a(u'') = v''$. By range-consistency, node $v'$ has to be inserted immediately after $v''$ in $\mathtt{LEX}_i$. If such a node $u''$ does not exist (i.e. $u'$ is the leftmost node in $\mathtt{LEX}_i$ having an outgoing edge labeled $a$), then $v'$ has to be inserted in $\mathtt{LEX}_i$ so that it becomes the first node with incoming edges labeled $a$ (i.e. in the position immediately following the rightmost node $v''$ with incoming label $a'$, where $a'$ is the lexicographically-largest character such that $a'\prec a$, or at the first position in $\mathtt{LEX}_i$ if such a character $a'$ does not exist).
This shows that invariant {\bf 2} is maintained: the Wheeler properties are kept true among nodes in  $\mathtt{LEX}_{i+1}$. It is also clear that we do not insert $\approx$-equivalent adjacent states with the same incoming label (see Figure \ref{fig:make WDAG}: by construction, the newly-inserted nodes $v_1, z'_3, v_2, z''_3, v_3$ are non-equivalent to their neighbors), i.e. invariant {\bf 3} is maintained. 
Finally, the accepted language does not change since the splitting process generates $\approx$-equivalent nodes: also invariant {\bf 1} stays true.

Note that the minimization process on the original acyclic DFA $\mathcal A$ takes linear time. After that, we only insert edges/nodes in the minimum output WDFA: never delete. 
It follows that the number of performed operations is equal to the output's size. 
The final automaton could be either smaller or exponentially-larger than $\mathcal A$.
We note that all the discussed operations can be easily implemented in logarithmic time using the data structures discussed in Section \ref{sec:data structures}: finding the $q$ runs of states $[u_{i_j}, u'_{i_j}]$, as well as finding the $O(k)$ states $w_i$, requires executing a constant number of \emph{rank} operations on sequence $\mathtt{OUT}$ and \emph{start} operations on $\mathtt{IN}$ for each predecessor of $v$.
Nodes can be inserted at the right position in sequence $\mathtt{LEX}$ exactly as done in Algorithm \ref{alg:step} (by also updating $\mathtt{IN}$ and $\mathtt{OUT}$). Finally, the graph can be dynamically updated (i.e. splitting nodes) and queried (i.e. navigation) by keeping it as a dynamic adjacency list: since we can spend logarithmic time per edge, we can store the graph as a self-balancing tree associating nodes to their predecessors and successors (also kept as self-balancing trees). This structure supports all updates and queries on the graph in logarithmic time. It follows that the overall procedure terminates in $O(n+m\log m)$ time, $n$ and $m$ being the input and output's sizes, respectively.

\begin{figure}
	\begin{center}
		\includegraphics[scale=0.5]{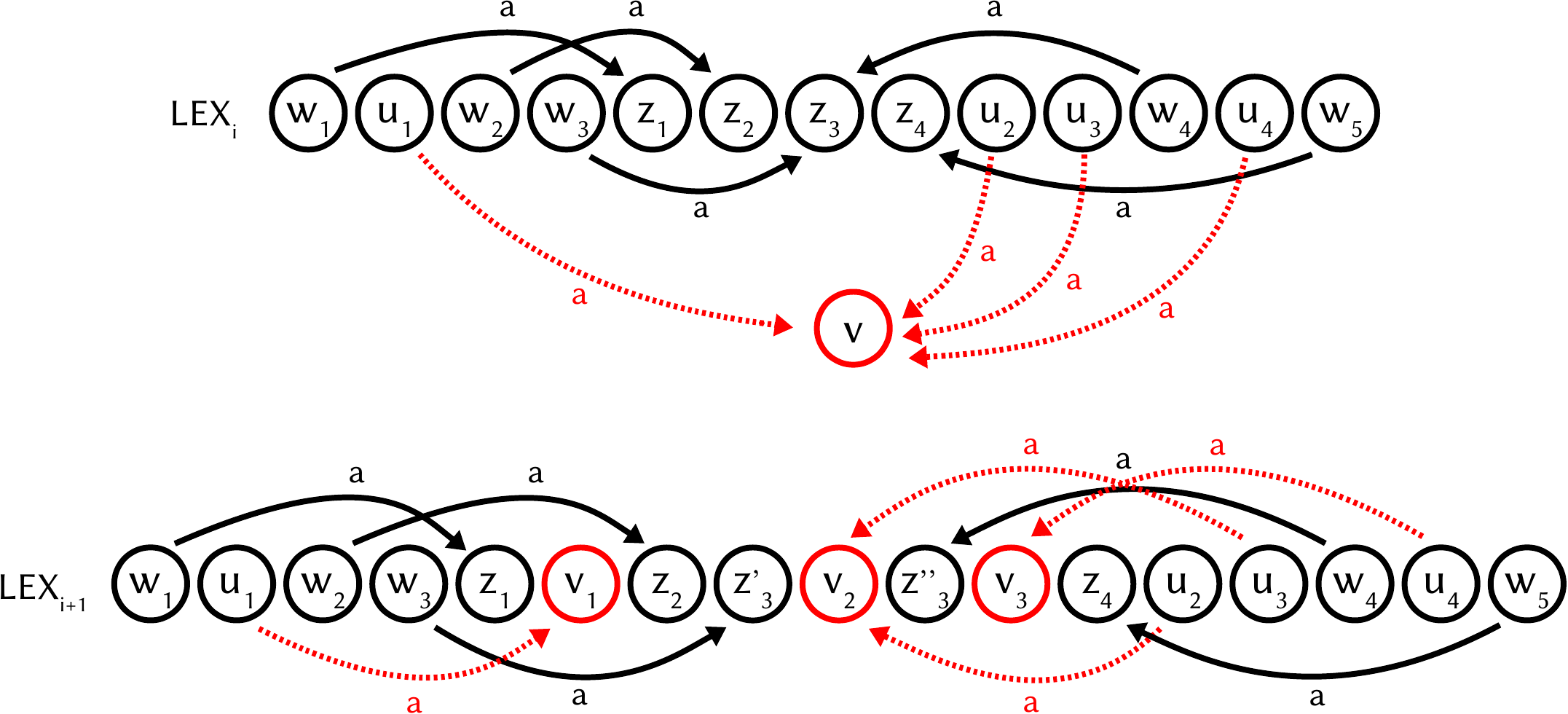}
	\end{center}
	\caption{Inconsistency resolution. 
		Nodes are ordered left-to-right by the total ordering $<$ (except $v$ in the top part of the figure).
		Top: we are trying to insert $v$ in $\mathtt{LEX}_i$, but this violates the Wheeler properties (edges' destinations are not ordered as the sources, no matter where we insert $v$). Bottom: we solve the inconsistencies by splitting $v$ in three equivalent nodes $v_1 \approx v_2 \approx v_3$ and $z_3$ in two equivalent nodes $z_3' \approx z_3''$. Note that (i) the splitting procedure induces naturally an ordering of the nodes that satisfies the Wheeler properties, and (ii) after splitting, no two adjacent states with the same incoming label are equivalent by $\approx$. By Theorem \ref{thm:characterization minimum}, this is the minimum way of splitting nodes. For simplicity, in the figure nodes $z_1,\dots, z_4$ do not coincide with any node $w_1, \dots, w_5$ or $u_1, \dots, u_4$. This may not necessarily be the case. In our full proof in Appendix \ref{app:DFA->WDFA} we show that our procedure is correct even when this happens.}\label{fig:make WDAG}
\end{figure}

\section{Worst case blowup from an acyclic DFA to the minimum equivalent WDFA}\label{sec:worst case blow up}

In this section we show that the running time of the conflict resolution algorithm in Section $\ref{sec:ADFA->WADFA}$ is exponential in the worst case, i.e. there exists a family of regular languages where the size of the smallest WDFA is exponential in the size of the smallest DFA. We now show that one such family the sequence of languages $L_1,L_2,\ldots$, where $L_m = \{ c \alpha e \; | \; \alpha \in \{a,b\}^m \} \cup \{ d \alpha f \; | \; \alpha \in \{a,b\}^m \}$.

For an example, Figure \ref{fig:dfa_wdfa_worst_case2} shows a DFA and the smallest WDFA for the language $L_3$. In general, we can build a DFA for $L_m$ by generalizing the construction in the figure: the source node has outgoing edges labeled with $c$ and $d$, followed by simple linear size "universal gadgets" capable of generating all binary strings of length $m$, with one gadget followed by an $e$ and the other by an $f$. The two sink states are the only accepting states.

The smallest WDFA for $L_m$ is an unraveling of the described DFA, such that all paths up to (but not including) the sinks end up in distinct nodes, i.e. the universal gadgets are replaced by full binary trees (see Figure \ref{fig:dfa_wdfa_worst_case2}). It is easy to see that the automaton is Wheeler as the only nodes that have multiple incoming paths are the sinks, and the sinks have unique labels.

To prove that this is the minimal WDFA, we need to check the condition of Theorem \ref{thm:characterization minimum}, i.e. that all colexicographically consecutive pairs of nodes with the same incoming label are Myhill-Nerode inequivalent. As labels $c,d,e$ and $f$ occur only once, it is enough to focus on nodes that have label $a$ or $b$. Let $B_1, B_2, B_{2^{m+1}-1}$ be the colexicographically sorted sequence of all possible binary strings with lengths $1 \leq |B_i| \leq m$ from the alphabet $\{a,b\}$. Observe that the nodes with incoming label $a$ and $b$ correspond to path labels of the form $c B_i$ and $dB_i$ for all $1 \leq i \leq 2^{m+1}-1$. The colexicographically sorted order of these path labels is:
$$c B_1 < d B_1 < c B_2 < d B_2 < \ldots < c B_{2^{m+1}-1} < d B_{2^{m+1}-1} $$
Here we can see that all consecutive pairs have a different first character, and therefore they lead to a different sink in the construction, and hence they are not Myhill-Nerode equivalent. Therefore the automaton is the minimum WDFA.

The DFA has $n = 4m + 5$ states and the WDFA has $1 + 2^{m+2} = 1 + 2^{(n-5)/4 + 2}$ states, so we obtain the following result:

\begin{theorem}
The minimal WDFA equivalent to an acyclic DFA with $n$ states has $\Omega(2^{n/4})$ states in the worst case.
\end{theorem}

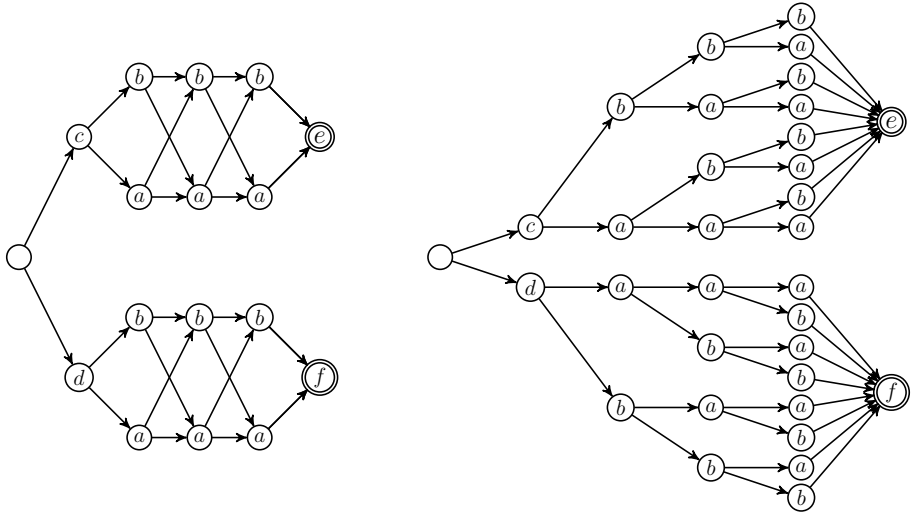
\begin{figure}

\centering

\begin{tikzpicture}[->,>=stealth', semithick, auto, scale=0.8]
  \tikzstyle{every state}=[scale=0.4]

 %\node[state, label=above:{$ c$}] (T22)    at (1.5,0)	    {}; 

 \node[state, accepting, label=above:{}] (T21)    at (-4,0+2)		    {\Huge $e$};
 
 %\node[state, label=above:{$ a$}] (T20) at (-1.5,-1) 		{};
 %\node[state, label=above:{$ b$}] (T19) at (-1.5,1)   		{};  
 
 \node[state, label=above:{}] (T18) at (-5,-1+2) 		{\huge $a$};
 \node[state, label=above:{}] (T17) at (-5,1+2)   		{\huge $b$};  
 
 \node[state, label=above:{}] (T16) at (-6,-1+2) 		{\huge $a$};
 \node[state, label=above:{}] (T15) at (-6,1+2)   		{\huge $b$};  

 \node[state, label=above:{}] (T14) at (-7,-1+2) 		{\huge $a$};
 \node[state, label=above:{}] (T13) at (-7,1+2)   		{\huge $b$};  
 
 \node[state, label=above:{}] (T12) at (-8,-0+2)   		{\huge $c$};

 %\node[state, label=above:{$ d$}] (T11)    at (1.5,-4)	    {};
 
 \node[state, accepting, label=above:{}] (T10)    at (-4,-4+2)	    {\huge $f$};
 
 %\node[state, label=above:{$ a$}] (T9) at (-1.5,-5) 		{};
 %\node[state, label=above:{$ b$}] (T8) at (-1.5,-3)   		{};  
 
 \node[state, label=above:{}] (T7) at (-5,-5+2) 		{\huge $a$};
 \node[state, label=above:{}] (T6) at (-5,-3+2)   		{\huge $b$};  
 
 \node[state, label=above:{}] (T5) at (-6,-5+2) 		{\huge $a$};
 \node[state, label=above:{}] (T4) at (-6,-3+2)   		{\huge $b$};  

 \node[state, label=above:{}] (T3) at (-7,-5+2)    		{\huge $a$};
 \node[state, label=above:{}] (T2) at (-7,-3+2)   		{\huge $b$};  

 \node[state, label=above:{}] (T1) at (-8,-4+2)   		{\huge $d$};

 \node[state, label=above:{}] (T0) at (-9,-2+2)   		{};  
 
 \draw (T0) edge node {} (T1);
 
 \draw (T1) edge node {} (T2);
 \draw (T1) edge node {} (T3);
 
 \draw (T2) edge node {} (T4);
 \draw (T2) edge node {} (T5);
 \draw (T3) edge node {} (T4);
 \draw (T3) edge node {} (T5); 

 \draw (T4) edge node {} (T6);
 \draw (T4) edge node {} (T7);
 \draw (T5) edge node {} (T6);
 \draw (T5) edge node {} (T7);  

 \draw (T6) edge node {} (T10);
 \draw (T6) edge node {} (T10);
 \draw (T7) edge node {} (T10);
 \draw (T7) edge node {} (T10);  
 
 %\draw (T8) edge node {} (T10);
 %\draw (T9) edge node {} (T10); 

 %\draw (T10) edge node {} (T11); 

 \draw (T0) edge node {} (T12);
 
 \draw (T12) edge node {} (T13);
 \draw (T12) edge node {} (T14);
 
 \draw (T13) edge node {} (T15);
 \draw (T13) edge node {} (T16);
 \draw (T14) edge node {} (T15);
 \draw (T14) edge node {} (T16); 

 \draw (T15) edge node {} (T17);
 \draw (T15) edge node {} (T18);
 \draw (T16) edge node {} (T17);
 \draw (T16) edge node {} (T18);  

 \draw (T17) edge node {} (T21);
 \draw (T17) edge node {} (T21);
 \draw (T18) edge node {} (T21);
 \draw (T18) edge node {} (T21);  
 
 %\draw (T19) edge node {} (T21);
 %\draw (T20) edge node {} (T21); 

 %\draw (T21) edge node {} (T22); 
 
 \node[state, label=above:{}] (S0) at (1-3,0)   		{};  
 
 \node[state, label=above:{}] (S11) at (2.5-3,1/2)   		{\huge $c$};
   \node[state, label=above:{}] (S21) at (4-3,1/2)   		{\huge $a$};
   \node[state, label=above:{}] (S22) at (4-3,5/2)   		{\huge $b$};
     \node[state, label=above:{}] (S31) at (5.5-3,1/2)   		{\huge $a$};
     \node[state, label=above:{}] (S32) at (5.5-3,3/2)   		{\huge $b$};
     \node[state, label=above:{}] (S33) at (5.5-3,5/2)   		{\huge $a$};
     \node[state, label=above:{}] (S34) at (5.5-3,7/2)   		{\huge $b$};
       \node[state, label=above:{}] (S41) at (7.0-3,1/2)   		{\huge $a$};
       \node[state, label=above:{}] (S42) at (7.0-3,2/2)   		{\huge $b$}; 
       \node[state, label=above:{}] (S43) at (7.0-3,3/2)   		{\huge $a$}; 
       \node[state, label=above:{}] (S44) at (7.0-3,4/2)   		{\huge $b$}; 
       \node[state, label=above:{}] (S45) at (7.0-3,5/2)   		{\huge $a$}; 
       \node[state, label=above:{}] (S46) at (7.0-3,6/2)   		{\huge $b$}; 
       \node[state, label=above:{}] (S47) at (7.0-3,7/2)   		{\huge $a$}; 
       \node[state, label=above:{}] (S48) at (7.0-3,8/2)   		{\huge $b$};
         \node[state, accepting, label=above:{}] (Se) at (8.5-3,4.5/2)   		{\Huge $e$};
 
  \node[state, label=above:{}] (D11) at (2.5-3,-1/2)   		{\huge $d$};
   \node[state, label=above:{}] (D21) at (4-3,-1/2)   		{\huge $a$};
   \node[state, label=above:{}] (D22) at (4-3,-5/2)   		{\huge $b$};
     \node[state, label=above:{}] (D31) at (5.5-3,-1/2)   		{\huge $a$};
     \node[state, label=above:{}] (D32) at (5.5-3,-3/2)   		{\huge $b$};
     \node[state, label=above:{}] (D33) at (5.5-3,-5/2)   		{\huge $a$};
     \node[state, label=above:{}] (D34) at (5.5-3,-7/2)   		{\huge $b$};
       \node[state, label=above:{}] (D41) at (7.0-3,-1/2)   		{\huge $a$};
       \node[state, label=above:{}] (D42) at (7.0-3,-2/2)   		{\huge $b$}; 
       \node[state, label=above:{}] (D43) at (7.0-3,-3/2)   		{\huge $a$}; 
       \node[state, label=above:{}] (D44) at (7.0-3,-4/2)   		{\huge $b$}; 
       \node[state, label=above:{}] (D45) at (7.0-3,-5/2)   		{\huge $a$}; 
       \node[state, label=above:{}] (D46) at (7.0-3,-6/2)   		{\huge $b$}; 
       \node[state, label=above:{}] (D47) at (7.0-3,-7/2)   		{\huge $a$}; 
       \node[state, label=above:{}] (D48) at (7.0-3,-8/2)   		{\huge $b$};
         \node[state, accepting, label=above:{}] (Df) at (8.5-3,-4.5/2)   		{\huge $f$};

\draw (S0) edge node {} (S11);
  \draw (S11) edge node {} (S21);
  \draw (S11) edge node {} (S22);
    \draw (S21) edge node {} (S31);
    \draw (S21) edge node {} (S32);
    \draw (S22) edge node {} (S33);
    \draw (S22) edge node {} (S34);
      \draw (S31) edge node {} (S41);
      \draw (S31) edge node {} (S42);
      \draw (S32) edge node {} (S43);
      \draw (S32) edge node {} (S44);
      \draw (S33) edge node {} (S45);
      \draw (S33) edge node {} (S46);
      \draw (S34) edge node {} (S47);
      \draw (S34) edge node {} (S48);
        \draw (S41) edge node {} (Se);
        \draw (S42) edge node {} (Se);
        \draw (S43) edge node {} (Se);
        \draw (S44) edge node {} (Se);
        \draw (S45) edge node {} (Se);
        \draw (S46) edge node {} (Se);
        \draw (S47) edge node {} (Se);
        \draw (S48) edge node {} (Se);
        
\draw (S0) edge node {} (D11);
  \draw (D11) edge node {} (D21);
  \draw (D11) edge node {} (D22);
    \draw (D21) edge node {} (D31);
    \draw (D21) edge node {} (D32);
    \draw (D22) edge node {} (D33);
    \draw (D22) edge node {} (D34);
      \draw (D31) edge node {} (D41);
      \draw (D31) edge node {} (D42);
      \draw (D32) edge node {} (D43);
      \draw (D32) edge node {} (D44);
      \draw (D33) edge node {} (D45);
      \draw (D33) edge node {} (D46);
      \draw (D34) edge node {} (D47);
      \draw (D34) edge node {} (D48);
        \draw (D41) edge node {} (Df);
        \draw (D42) edge node {} (Df);
        \draw (D43) edge node {} (Df);
        \draw (D44) edge node {} (Df);
        \draw (D45) edge node {} (Df);
        \draw (D46) edge node {} (Df);
        \draw (D47) edge node {} (Df);
        \draw (D48) edge node {} (Df);

\end{tikzpicture}

\caption{Left: a DFA recognizing $L_3$. Right: the minimal WDFA recognizing $L_3$. For clarity the labels are drawn on the nodes: the label of an edge is the label of the destination node.}
\label{fig:dfa_wdfa_worst_case2}

\end{figure}

\bibliographystyle{plainurl}% the mandatory bibstyle
\bibliography{recognition}

\end{document}